\documentclass[11pt]{article}

\usepackage{osid2}
\usepackage{graphicx}
\usepackage{amsmath}
\usepackage{latexsym}
\usepackage{amsfonts}
\usepackage{amssymb}
\usepackage[english]{babel}

\newenvironment{proof}{\trivlist
\item[\hskip5pt\uppercase{Proof}.]\rm\hskip3pt}
{\hfill \rule{5pt}{5pt} \endtrivlist}

\newcommand{\frecc}{\to} 

\newcommand{\R}{\mathbb R} % reali
\newcommand{\C}{{\mathbb C}} % complessi
\newcommand{\N}{\mathbb N} % naturali

\newcommand{\hotimes}{\hat{\otimes}}% prodotto tensore algebrico
\newcommand{\votimes}{\bar{\otimes}}% prodotto tensore di algebre di v.N.
\newcommand{\hagotimes}{\otimes_{\rm h}}% Haagerup tensor product
\newcommand{\whotimes}{\otimes_{\rm w^\ast h}}% weak*-Haagerup tensor product
\renewcommand{\aa}{\mathcal{A}} % spazio di Hilbert
\newcommand{\bb}{\mathcal{B}} % spazio di Hilbert
\newcommand{\hh}{\mathcal{H}} % spazio di Hilbert H
\newcommand{\hhat}{\hat{\hh}} % spazio di Hilbert H^
\newcommand{\uhat}{{\hat{\uu}}} % spazio di Hilbert U^
\newcommand{\hhn}{{\mathcal{H}^{(n)}}} % spazio di Hilbert Hn
\newcommand{\kk}{\mathcal{K}} % spazio di Hilbert K
\newcommand{\kkn}{{\mathcal{K}^{(n)}}} % spazio di Hilbert Kn
\newcommand{\mm}{\mathcal{M}} % algebra di v.N. M
\newcommand{\mmp}{{\mathcal{M}^\prime}} % commutante dell'algebra di v.N. M
\newcommand{\nn}{\mathcal{N}} % algebra di v.N. N
\newcommand{\mmn}{{\mathcal{M}^{(n)}}} % algebra di v.N. Mn
\newcommand{\nnn}{{\mathcal{N}^{(n)}}} % algebra di v.N. Nn
\newcommand{\rr}{\mathcal{R}} % radicale
\newcommand{\uu}{\mathcal{U}} % spazio di Hilbert U
\newcommand{\vv}{\mathcal{V}} % spazio di Hilbert V
\newcommand{\ww}{\mathcal{W}} % spazio di Hilbert W
\newcommand{\ee}{\mathcal{E}} % mappa CP
\newcommand{\ff}{\mathcal{F}} % mappa CP
\renewcommand{\gg}{\mathcal{G}} % mappa CP
\renewcommand{\ss}{\mathcal{S}} % spazio denso
\renewcommand{\SS}{\mathsf{S}} % supermappa CP
\newcommand{\TT}{\mathsf{T}} % supermappa CP
\newcommand{\II}{\mathsf{I}} % supermappa CP
\newcommand{\CC}{\mathsf{C}} % supermappa CP
\newcommand{\PI}{\mathsf{\Pi}} % supermappa CP
\newcommand{\lh}{\mathcal{L}(\hh)} % operatori limitati su H
\newcommand{\lzh}{\mathcal{L}_0 (\hh)} % operatori compatti su H
\newcommand{\lk}{\mathcal{L}(\kk)} % operatori limitati su K
\newcommand{\lu}{\mathcal{L}(\uu)} % operatori limitati su U
\newcommand{\lv}{\mathcal{L}(\vv)} % operatori limitati su V
\newcommand{\elle}[1]{\mathcal{L} ( #1 )} % operatori limitati in spazi di Banach generici
\newcommand{\trcl}[1]{\mathcal{T} \left( #1 \right)} % operatori di classe traccia in spazi di Hilbert generici
\newcommand{\spanno}[1]{{\rm span}\, \left\{ #1 \right\}} % span lineare
\newcommand{\spannochiuso}[1]{\overline{\rm span}\, \left\{ #1 \right\}} % span lineare chiuso
\newcommand{\scal}[2]{\left\langle #1 , #2 \right\rangle} % prodotto scalare
\newcommand{\no}[1]{\left\|#1\right\|} % norma
\newcommand{\trt}[1]{{\rm tr} \left( #1 \right)} % traccia con parentesi tonda
\newcommand{\trq}[1]{{\rm tr} \left[ #1 \right]} % traccia con parentesi quadra
\newcommand{\cb}[1]{{\rm CB} (#1)} % mappe completamente limitate e wk-* continue
\newcommand{\cp}[1]{{\rm CP} (#1)} % mappe completamente positive normali
\newcommand{\cpn}[1]{{\rm CP}_1 (#1)} % quantum channel
\newcommand{\cpq}[1]{{\rm SCP} \left(#1\right)} % mappe completamente positive su CB
\newcommand{\cpqn}[1]{{\rm SCP}_1 \left(#1\right)} % supermappe deterministiche
\newcommand{\ii}{\mathcal{I}} % identita' su L(H)
\newcommand{\jj}{\mathcal{J}} % strumento

\newcommand{\wklim}{\operatornamewithlimits{\rm wk*-lim}} % limite nella topologia weak-*

\newcommand{\lam}{\lambda}
\newcommand{\Lam}{\Lambda}

\title{Normal Completely Positive Maps on the Space of Quantum Operations}
\author{Giulio Chiribella\\{\footnotesize\it Center for Quantum Information, Institute for Interdisciplinary Information Sciences, Tsinghua University, Beijing 100084, China\\
E-mail: gchiribella@mail.tsinghua.edu.cn}\\[2ex]
Alessandro Toigo\\{\footnotesize\it Dipartimento di Matematica, Politecnico di Milano, Piazza Leonardo da Vinci 32, I-20133 Milano, Italy\\
Istituto Nazionale di Fisica Nucleare, Sezione di Milano, Via Celoria 16, I-20133 Milano, Italy\\
E-mail: alessandro.toigo@polimi.it}\\[2ex]
Veronica Umanit\`a\\{\footnotesize\it Dipartimento di Matematica, Universit\`a di Genova, Via Dodecaneso 35, I-16146 Genova, Italy\\
E-mail: umanita@dima.unige.it} }

\begin{document}

\maketitle
\begin{abstract}
Quantum supermaps are higher-order maps transforming quantum operations into quantum operations.  
Here we extend the theory of quantum supermaps, originally formulated in the finite dimensional setting, to the case of higher-order  maps transforming quantum operations with input in a separable von Neumann algebra and output in the algebra of the bounded operators on a given separable Hilbert space.   
In this setting we prove two dilation theorems for quantum supermaps that are the analogues of the Stinespring and Radon-Nikodym theorems for quantum operations. 
Finally, we  consider the case of quantum superinstruments, namely measures with values in the set of quantum supermaps, and derive a dilation theorem for them that is analogue to Ozawa's theorem for quantum instruments. The three dilation theorems presented  here show that all the supermaps defined in this paper   can be implemented by connecting devices in quantum circuits.
\end{abstract}

\section{Introduction}

{\em Quantum supermaps} \cite{CDaP1,CDaP2} are the most general admissible transformations of quantum devices. Mathematically, the action of a quantum device is associated to a set of  completely positive trace non-increasing maps, called  {\em quantum operations} \cite{QTOS76,Kraus71}, which  transform the states of an input quantum system into states of an output quantum system.
In the dual (Heisenberg) picture, quantum operations are given by  normal completely positive maps transforming the observables of the output system into observables of the input system, with the condition that each quantum operation is upper bounded by a unital completely positive map.  
A quantum supermap is then a higher-order linear map that transforms quantum operations into quantum operations.

The theory of quantum supermaps has proven to be a powerful tool for the treatment of many advanced topics in quantum information theory \cite{opttomo,tredeoff,memorydisc,procqcmc,SedlakZiman,Ziman}, including in particular the optimal cloning and the optimal learning of unitary transformations \cite{unitlearn,unitclon} and quantum measurements \cite{measclon,measlearn}. Moreoever, quantum supermaps are interesting for the foundations of Quantum Mechanics as they are the possible dynamics in a toy model of non-causal theory \cite{puri}. A particular type of quantum supermaps has been considered by Zyczkowski \cite{zycko}, who used them to construct a theory with a state space that has a quartic relation between the number of distinguishable states  and the number of parameters needed to specify a state. Quantum supermaps also attracted interest in the mathematical physics literature, as they suggested the study of a general class of completely positive maps between convex subsets of the state space \cite{jencova}.    

Originally, the definition and the main theorems on quantum supermaps were presented by D'Ariano, Perinotti, and one of the authors in the context of full matrix algebras describing finite dimensional quantum systems \cite{CDaP1,CDaP2}.  An extension of the theory that includes both classical and quantum systems has been exposed informally in \cite{bitcommitment}, still in the finite dimensional setting. However, a rigorous definition and characterization of quantum supermaps in infinite dimension and for arbitrary von Neumann algebras is still lacking. This problem will be the main focus of the present paper.

Before presenting our results, we briefly review the definition and characterization of supermaps for full matrix algebras. Quantum supermaps are defined axiomatically as linear completely positive maps transforming quantum operations into quantum operations (see \cite{CDaP1,CDaP2} for the physical motivation of linearity and complete positivity). A quantum supermap is  {\em deterministic} if it transforms quantum channels (i.e.~unital completely positive maps, see e.g.~\cite{Holevo01}) into quantum channels.
References \cite{CDaP1,CDaP2} proved the following dilation theorem for deterministic supermaps: denoting by $\lh$ and $\lk$ the $C^\ast$-algebras of linear operators on the finite dimensional Hilbert spaces $\hh$ and $\kk$, respectively, and writing $\cp{\lh,\lk}$ for the set of completely positive maps sending $\lh$ into $\lk$, we have that any deterministic supermap $\SS$ transforming quantum operations in $\cp{\elle{\hh_1}, \elle{\kk_1}}$ to quantum operations in $\cp{\elle{\hh_2}, \elle{\kk_2}}$ has the following form:  
\begin{equation}\label{eq. intro 1}
[\SS (\ee)] (A) = V_1^\ast \left[(\ee\otimes \ii_{\vv_1}) ( V_2^\ast (A\otimes I_{\vv_2}) V_2 ) \right] V_1 \quad \forall A\in \elle{\hh_2}
\end{equation}
for all $\ee\in\cp{\elle{\hh_1}, \elle{\kk_1}}$, where  $\vv_1$ and $\vv_2$ are two ancillary finite dimensional Hilbert spaces, $V_1:  \kk_2 \to  \kk_1  \otimes \vv_1$  and $V_2 :   \hh_1  \otimes \vv_1  \to  \hh_2 \otimes \vv_2 $ are isometries, $\ii_{\vv_1}$ is the identity map on $\elle{\vv_1}$ and $I_{\vv_2}$ is the identity operator on $\vv_2$.
In the Schr\"odinger (or predual) picture, this result shows that the most general way to transform a quantum operation is achieved by connecting the corresponding device in a quantum circuit   consisting in the following sequence of operations:  
\begin{enumerate}
\item apply an invertible transformation (corresponding to the isometry $V_1$), which transforms  the system $\kk_2$ into the composite system $\kk_1 \otimes \vv_1$;
\item use the input device on system $\kk_1$, thus transforming it into system $\hh_1$;  in the Schr\"odinger picture the action of the device will correspond to a set of predual quantum operations $\ee_\ast$ transforming states on $\kk_1$ into states on $\hh_1$;
\item apply an invertible transformation (corresponding to the isometry $V_2$),  which transforms  the composite system $\hh_1 \otimes \vv_1$ into the composite system $\hh_2 \otimes \vv_2$;
\item discard system $\vv_2$ (mathematically,  take the partial trace over $\vv_2$).
\end{enumerate}

In this paper we will extend Eq.~\eqref{eq. intro 1} and the other results of \cite{CDaP1,CDaP2,bitcommitment} to the case where the input spaces $\elle{\hh_i}$ of the quantum operations are replaced by arbitrary separable von Neumann algebras and the outputs $\elle{\kk_i}$ also are allowed to be infinite dimensional. The usefulness of this extension for applications is twofold: on the one hand, it removes the restriction to finite dimensional quantum systems and provides the natural generalization of quantum supermaps to the infinite dimensional case; on the other hand, replacing the input algebras $\elle{\hh_i}$ with generic separable von Neumann algebras, it allows us to include transformations of quantum measuring devices, which are described in the Schr\"odinger picture by maps from the algebra of bounded operators on the Hilbert space of the measured system to the commutative algebra of functions on the outcome space. The  supermaps defined in this paper are thus able to describe tasks like `measuring a measurement' \cite{lor,fiur,soto,eisert}, where one tries to measure properties of a quantum measuring device by inserting it in a suitable circuit.

In trying to extend Eq.~\eqref{eq. intro 1} to the infinite dimensional setting, one encounters two key differences with respect to the finite dimensional case.  The first difference concerns the domain of definition of quantum supermaps. Clearly, the natural domain for a quantum supermap is the linear space spanned by quantum operations. However, while in finite dimensions quantum operations in $\cp{\elle{\hh_i},\elle{\kk_i}}$ span the whole set of linear maps from $\elle{\hh_i}$ to $\elle{\kk_i}$, in infinite dimension they only span the {\em proper} subset $\cb{\elle{\hh_i},\elle{\kk_i}}$ of weak*-continuous completely bounded linear maps, which is even smaller than the set of bounded linear maps from $\elle{\hh_i}$ to $\elle{\kk_i}$. The second key difference concerns the necessary and sufficient conditions needed for the proof of the dilation theorem.  Indeed, not every deterministic quantum supermap admits a dilation of the form of Eq.~\eqref{eq. intro 1} in the infinite dimensional case.  We will prove that such a dilation exists if and only if  the deterministic supermap $\SS$ is  \emph{normal}, in a suitable sense that will be defined later.    Under the normality hypothesis, a natural algebraic construction leads to our dilation theorem (Theorem \ref{teo. centr.}) for deterministic supermaps, which is the main result of the paper. This result can be compared with analogous results in the theory of operator spaces \cite{BlM,Paul}, which, however, though being more general, are much less specialized than ours and imply it only in trivial cases (see Remark \ref{rem:opspa2} in Section \ref{sez. stine} for a brief discussion).

Our second result is a Radon-Nikodym theorem for probabilistic supermaps, namely supermaps that are dominated by deterministic supermaps.  The class of probabilistic supermaps is particularly interesting for physical applications, as such maps naturally appear in the description of quantum circuits that are designed to test properties of physical devices \cite{combs,CDaP1,CDaP2,bitcommitment}. Higher-order quantum measurements are indeed described by \emph{quantum superinstruments}, which are the generalization of the quantum instruments of Davies and Lewis \cite{DavLew}. The third main result of the paper then will be  the proof of a dilation theorem for quantum superinstruments, in analogy with Ozawa's dilation theorem for ordinary instruments \cite{Ozawa}.

The paper is organized as follows. In Section \ref{sez. notaz.} we fix the elementary definitions and notations, and state or recall some basic facts needed in the rest of the paper. In particular, in Section \ref{sez. succ. cresc.} we extend the notion of increasing nets from positive operators to normal completely positive maps, while Section \ref{sez. amplificaz.} contains  some elementary results about the tensor product of weak*-continuous completely bounded maps. In Section \ref{sez. centr.} we define normal completely positive supermaps  and provide some examples.   In Section \ref{sez. stine} we prove the dilation Theorem \ref{teo. centr.} for deterministic supermaps.  As an application of Theorem \ref{teo. centr.}, in Section \ref{subsect:meastochan} we show that every deterministic supermap transforming measurements into quantum operations can be realized by connecting devices in a quantum circuit.   Section \ref{sez. Radon}  extends Theorem \ref{teo. centr.} to probabilistic supermaps, providing a Radon-Nikodym theorem for supermaps. We then define quantum superinstruments in Section \ref{sez. superstr.} and use the Radon-Nikodym theorem to prove a dilation theorem for quantum superinstruments, in  analogy with Ozawa's result for ordinary instruments  (see in particular Proposition 4.2 in \cite{Ozawa}). The dilation theorem for quantum superinstruments is finally applied in Section \ref{subsect:measmeas} to show how every abstract superinstrument describing a measurement on a quantum measuring device can be realized in a circuit.

\section{Preliminaries and notations}\label{sez. notaz.}

In this paper, unless the contrary is  explicitly stated, we will always mean by \emph{Hilbert space} a complex and separable Hilbert space, with norm $\no{\cdot}$ and scalar product $\langle\cdot,\cdot\rangle$ linear in the second entry. If $\hh$, $\kk$ are Hilbert spaces, we denote by $\elle{\hh,\kk}$ the Banach space of bounded linear operators from $\hh$ to $\kk$ endowed with the uniform norm $\no{\cdot}_\infty$. If $\hh=\kk$, we will use the shortened notation $\lh :=\elle{\hh,\hh}$, and $I_\hh$ will be the identity operator in $\lh$. The linear space $\lh$ is ordered in the usual way by the cone of positive (semidefinite) operators. We denote by $\leq$ the order relation in $\lh$, and by $\elle{\hh}_+$ the cone of positive operators.

By {\em von Neumann algebra} we mean a $\ast$-subalgebra $\mm\subset\lh$ such that $\mm=(\mm')'$, where $\mm'$ denotes the commutant of $\mm$ in $\lh$. Note that, as we will always assume  that the Hilbert space $\hh$ is separable, the von Neumann algebras  considered here will be those called \emph{separable} in the literature.

When $\mm$ is regarded as an abstract von Neumann algebra (i.e.~without reference to the representing Hilbert space $\hh$), we will write its identity element $I_\mm$ instead of $I_\hh$. As usual, we define $\mm_+ := \mm\cap\elle{\hh}_+$. The identity map on $\mm$ will be denoted by $\ii_\mm$, and, when $\mm\equiv\lh$, the abbreviated notation $\ii_\hh:=\ii_{\elle{\hh}}$ will be used.

The algebraic tensor product of linear spaces $U$, $V$ will be written $U\hotimes V$, while the notation $\hh\otimes\kk$ will be reserved to denote the Hilbert space tensor product of the Hilbert spaces $\hh$ and $\kk$.
The inclusion $\hh\hotimes\kk \subset \hh\otimes\kk$ holds, and it is actually an equality iff $\hh$ or $\kk$ is finite dimensional. We will sometimes use the notation $\hhn : = \C^n \otimes \hh$.

If $A\in\lh$ and $B\in\lk$, their tensor product $A\otimes B$, which is well defined as a linear map on $\hh\hotimes\kk$, uniquely extends to a bounded operator $A\otimes B \in\elle{\hh\otimes\kk}$ in the usual way (see e.g.~p.~183 in \cite{Tak}). Thus, the algebraic tensor product $\lh\hotimes\lk$ can be regarded as a linear subspace of $\elle{\hh\otimes\kk}$. Also in this case, the equality $\lh\hotimes\lk = \elle{\hh\otimes\kk}$ holds iff $\hh$ or $\kk$ is finite dimensional. More generally, let $\mm\subset \lh$ and $\nn\subset \lk$ be two von Neumann algebras. Then, $\mm\hotimes\nn$ is a linear subspace of $\elle{\hh\otimes\kk}$. Its weak*-closure is the von Neumann algebra $\mm\votimes\nn \subset \elle{\hh\otimes\kk}$ (see Definition 1.3 p.~183 in \cite{Tak}). Clearly, $\mm\hotimes\nn = \mm\votimes\nn$ iff $\mm$ or $\nn$ is finite dimensional. It is a standard fact that $\lh\votimes\lk = \elle{\hh\otimes\kk}$ (see Eq.~10, p.~185 in \cite{Tak}).

We denote by $M_n (\C)$ the linear space of square $n\times n$ complex matrices, which we identify as usual with the space $\elle{\C^n}$. If $\mm\subset\lh$ is a von Neumann algebra, we write $\mmn : = M_n (\C) \votimes \mm$, which is a von Neumann algebra contained in $\elle{\hhn}$. As remarked above, $\mmn$ coincides with the algebraic tensor product $M_n (\C) \hotimes \mm$. If $\ee : M_m (\C) \frecc M_n (\C)$ and $\ff : \mm \frecc \nn$ are linear operators, we then see that their algebraic tensor product can be regarded as a linear map $\ee \otimes \ff : \mm^{(m)} \frecc \nnn$. Since both $\mm^{(m)}$ and $\nnn$ are von Neumann algebras, it makes sense to speak about positivity and boundedness of $\ee \otimes \ff$. This fact is at the heart of the  classical definitions of complete positivity and complete boundedness. In both  definitions, we use $\ii_n$ to denote the identity map on $M_n (\C)$, i.e.~$\ii_n : = \ii_{M_n (\C)}$.

\begin{definition}{Definition}\label{def:CB-CP}
Let $\mm$, $\nn$ be two von Neumann algebras. Then a linear map $\ee : \mm \frecc \nn$ is
\begin{itemize}
\item[-] {\em completely positive (CP)} if the linear map $\ii_n \otimes \ee$ is positive, i.e.~maps $\mmn_+$ into $\nnn_+$, for all $n\in\N$;
\item[-] {\em completely bounded (CB)} if there exists $C>0$ such that, for all $n\in\N$,
$$
\|(\ii_n \otimes \ee)(\tilde{A})\|_\infty \leq C \|\tilde{A}\|_\infty \quad \forall \tilde{A}\in \mmn ,
$$
i.e.~if the linear map $\ii_n \otimes \ee$ is bounded from the Banach space $\mmn$ into the Banach space $\nnn$ for all $n\in\N$, and the uniform norms of all the maps $\{\ii_n \otimes \ee\}_{n\in\N}$ are majorized by a constant independent of $n$.
\end{itemize}
\end{definition}

\begin{definition}{Example}
The simplest example of CP and CB map is given by a $\ast$-homomorphism $\pi: \mm \frecc \nn$. Indeed, for all $n\in\N$ the tensor product $\ii_n \otimes \pi : \mmn \frecc \nnn$ is again a $\ast$-homomorphism, hence it is positive and satisfies $\|(\ii_n \otimes \pi)(\tilde{A})\|_\infty \leq \|\tilde{A}\|_\infty \ \forall \tilde{A}\in\mmn$.
\end{definition}

We recall that a positive linear map $\ee:\mm \frecc \nn$ is {\em normal} if it preserves the limits of increasing and bounded sequences, i.e.~$\ee(A_n) \uparrow \ee(A)$ in $\nn$ for all increasing sequences $\{A_n\}_{n\in\N}$ and $A$ in $\mm_+$ such that $A_n\uparrow A$ (as usual, the notation $A_n\uparrow A$ means that $A$ is the {\em least upper bound} of the sequence $\{A_n\}_{n\in\N}$ in $\mm$, see e.g.~Lemma 1.7.4 in \cite{Sakai}). It is a standard fact that a positive linear map $\ee:\mm \frecc \nn$ is normal if and only if it is weak*-continuous (Theorem 1.13.2 in \cite{Sakai}). 

We introduce the following notations:
\begin{itemize}
\item[-] $\cb{\mm,\nn}$ is the linear space of {\em weak*-continuous} CB maps from $\mm$ to $\nn$;
\item[-] $\cp{\mm,\nn}$ is the set of \emph{normal} CP maps from $\mm$ to $\nn$;
\item[-] $\cpn{\mm,\nn}$ is the set of \emph{quantum channels} from $\mm$ to $\nn$, i.e.~the subset of elements $\ee\in\cp{\mm,\nn}$ such that $\ee(I_\mm) = I_\nn$.
\end{itemize}

\begin{definition}{Remark}\label{rem:CB=Lin in dim finita}
Suppose $\nn\subset M_n (\C)$. Then the set $\cb{\mm,\nn}$ coincides with the space of all weak*-continuous linear maps from $\mm$ to $\nn$ (see e.g.~Exercise 3.11 in \cite{Paul}). In particular, if also $\mm\subset M_m (\C)$, then $\cb{\mm,\nn}$ is the set of all linear maps from $\mm$ to $\nn$.
\end{definition}

If $\mm_1$, $\mm_2$ and $\mm_3$ are von Neumann algebras and $\ff\in\cb{\mm_1,\mm_2}$, $\ee\in\cb{\mm_2,\mm_3}$, then $\ee\ff\in\cb{\mm_1,\mm_3}$. The same fact is true if we replace all ${\rm CB}$ spaces with ${\rm CP}$'s or ${\rm CP}_1$'s.    

\begin{definition}{Remark}\label{restr-CP}
Let $\mm_0$, $\mm$ be two von Neumann algebras contained in the same operator algebra $\lh$, with $\mm_0\subset \mm$. Since the inclusion map $\ii_{\mm_0,\mm} : \mm_0 \hookrightarrow \mm$ is in $\cpn{\mm_0,\mm}$, it follows by the composition property that the restriction $\ee \mapsto \left.\ee\right|_{\mm_0} = \ee \ii_{\mm_0,\mm}$ maps $\cb{\mm,\nn}$ [resp., $\cp{\mm,\nn}$; $\cpn{\mm,\nn}$] into $\cb{\mm_0,\nn}$ [resp., $\cp{\mm_0,\nn}$; $\cpn{\mm_0,\nn}$]. A similar application of the composition property also shows the inclusions $\cb{\nn,\mm_0} \hookrightarrow \cb{\nn,\mm}$, $\cp{\nn,\mm_0} \hookrightarrow \cp{\nn,\mm}$ and $\cpn{\nn,\mm_0} \hookrightarrow \cpn{\nn,\mm}$. 
\end{definition}

The relation between the two sets $\cb{\mm,\nn}$ and $\cp{\mm,\nn}$ is shown in the following theorem (see also \cite{Haag}).
\begin{theorem}{Theorem}\label{CB = span CP}
The inclusion $\cp{\mm , \nn} \subset \cb{\mm , \nn}$ holds, and the set $\cp{\mm , \nn}$ is a cone in the linear space $\cb{\mm , \nn}$. For $\nn\equiv\lk$, the linear space spanned by $\cp{\mm , \lk}$ coincides with $\cb{\mm , \lk}$. More precisely, if $\ee\in\cb{\mm , \lk}$, then there exists four maps $\ee_k\in\cp{\mm , \lk}$ ($k=0,1,2,3$) such that $\ee = \sum_{k=0}^3 i^k \ee_k$.
\end{theorem}
\begin{proof}
We have already remarked that, if a positive map $\ee : \mm \frecc \nn$ is normal, then it is weak*-continuous. If $\ee$ is CP, then it is CB by Proposition 3.6 in \cite{Paul}. Thus, the inclusion $\cp{\mm , \nn} \subset \cb{\mm , \nn}$ holds. Clearly, $\cp{\mm , \nn}$ is a cone in $\cb{\mm , \nn}$.

Now, suppose $\ee\in\cb{\mm , \lk}$. By Theorem 8.4 in \cite{Paul}, there exists a (not necessarily separable) Hilbert space $\hhat$, a unital $\ast$-homomorphism $\pi : \mm \frecc \elle{\hhat}$ and bounded operators $V_i : \kk\frecc\hhat$ ($i=1,2$) such that
$$
\ee(A) = V_1^\ast \pi(A) V_2 \quad \forall A\in\mm \, .
$$
Let $\mm^\ast$ be the Banach dual space of $\mm$, and let $\mm^\ast = \mm_\ast\oplus\mm_\ast^\perp$ be the direct sum decomposition of $\mm^\ast$ into its normal and singular parts, as described in Definition 2.13 p.~127 of \cite{Tak} (the normal part $\mm_\ast$ coincides with the {\em predual} of $\mm$). If $u,v\in\kk$ [resp., $u,v\in\hhat$], denote by $\omega_{u,v}$ the element in the Banach dual $\lk^\ast$ of $\lk$ [resp., $\elle{\hhat}^\ast$ of $\elle{\hhat}$] given by $\omega_{u,v} (A) = \scal{u}{Av}$ for all $A\in\lk$ [resp., $A\in\elle{\hhat}$]. By Theorem 2.14 p.~127 in \cite{Tak}, there exists an orthogonal projection $P\in\elle{\hhat}$ which commutes with $\pi$ and is such that:
\begin{itemize}
\item[-] the $\ast$-homomorphism $A\mapsto \left.\pi(A)\right|_{P\hhat}$ is a normal representation of $\mm$ on $P\hhat$;
\item[-] $^t \pi (\omega_{P^\perp u,P^\perp v}) \in \mm_\ast^\perp$ for all $u,v\in\hhat$, where $P^\perp = I_{\hhat} - P$ and  $^t \pi$ is the transpose of $\pi$, defined by $[^t \pi(\omega)] (A)  :  =  \omega (\pi (A)) \ \forall \omega \in \elle{\hhat}^\ast, \  A \in \mm$.
\end{itemize}
Since $P$ and $\pi$ commute, we have
$$
^t \ee(\omega_{u,v}) = \,^t \pi (\omega_{V_1 u,V_2 v}) = \,^t \pi (\omega_{PV_1 u,PV_2 v}) + \,^t \pi (\omega_{P^\perp V_1 u,P^\perp V_2 v}) \, .
$$
Since $^t \ee(\omega_{u,v}) \in \mm_\ast$, $\,^t \pi (\omega_{PV_1 u,PV_2 v}) \in \mm_\ast$ and $\,^t \pi (\omega_{P^\perp V_1 u,P^\perp V_2 v}) \in \mm_\ast^\perp$, it follows that $\,^t \pi (\omega_{P^\perp V_1 u,P^\perp V_2 v}) = 0$, hence
$$
^t \ee(\omega_{u,v}) = \,^t \pi (\omega_{PV_1 u,PV_2 v}) \quad \forall u,v\in\kk
$$
or, equivelently,
$$
\scal{u}{\ee(A)v} = \scal{PV_1 u}{\pi(A) PV_2 v} \quad \forall A\in\mm \, , \, u,v\in\kk \, .
$$
We thus see that $\ee(A) = V_1^\ast P \pi(A) PV_2$ for all $A\in\mm$. As $\pi$ restricted to the subspace $P\hhat$ is normal, then each map $\ee_k$ ($k=0,1,2,3$), given by
$$
\ee_k (A) = \frac{1}{4} (i^k V_1 + V_2)^\ast P \pi(A) P (i^k V_1 + V_2) \quad \forall A\in\mm \, ,
$$
is in $\cp{\mm,\lk}$. Since $\ee = \sum_{k=0}^3 i^k \ee_k$, this shows that $\cb{\mm,\lk}$ is the linear span of $\cp{\mm,\lk}$.
\end{proof}

The cone $\cp{\mm, \nn}$ induces a linear ordering in the space $\cb{\mm, \nn}$, that we will denote by $\preceq$. Namely, given two maps $\ee,\ff \in \cb{\mm, \nn}$, we will write $\ee \preceq  \ff$  whenever $\ff - \ee \in\cp{\mm, \nn}$.

An elementary example of maps in $\cb{\mm,\lk}$ is constructed in the following way. Suppose $\mm\subset\lh$. For $E\in\elle{\hh, \kk}$, $F\in\elle{\kk, \hh}$, denote by $E\odot_\mm F$ the linear map
$$
E\odot_\mm F : \mm \frecc \lk \, , \qquad (E\odot_\mm F)(A) = E AF 
$$
[Note that the domain of the map $E\odot_\mm F$ is explicitly indicated by the subscript $\mm $].

The main properties of $E\odot_\mm F$ are collected in the next proposition.
\begin{theorem}{Proposition}\label{prop:prop. di odot}
Suppose $\mm\subset\lh$ is a von Neumann algebra. Then, for all $E\in\elle{\hh, \kk}$ and $F\in\elle{\kk, \hh}$,
\begin{enumerate}
\item $E\odot_\mm F\in\cb{\mm,\lk}$;
\item $F^\ast\odot_\mm F\in\cp{\mm,\lk}$;
\item for all operators $A\in\lh$ in the commutant $\mmp$ of $\mm$, we have $EA\odot_\mm F = E\odot_\mm AF$;
\item for all $A\in\mmp$, with $0\leq A\leq I_\hh$, we have $A^{\frac 12}\odot_\mm A^{\frac 12} \preceq \ii_{\mm,\lh}$, where the map $\ii_{\mm,\lh}$ is the inclusion $\mm\hookrightarrow\lh$.
\end{enumerate}
\end{theorem}
\begin{proof}
(1) Weak*-continuity of $E\odot_\mm F$ is clear. For all $n\in \N$, we have the equality $\ii_n \otimes (E\odot_\mm F) = (I_{\C^n}\otimes E)\odot_{\mm^{(n)}} (I_{\C^n}\otimes F)$, and then
\begin{align*}
\|[\ii_n \otimes (E\odot_\mm F)](\tilde{A})\|_\infty & \leq \no{I_{\C^n}\otimes E}_\infty \|\tilde{A}\|_\infty \no{I_{\C^n}\otimes F}_\infty \\
& = \no{E}_\infty \|\tilde{A}\|_\infty \no{F}_\infty
\end{align*}
for all $\tilde{A}\in\mmn$, which shows that $E\odot_\mm F$ is CB (with $C=\no{E}_\infty \no{F}_\infty$).

(2) For all $n\in \N$, we have $\ii_n \otimes (F^\ast\odot_\mm F) = (I_{\C^n}\otimes F)^\ast \odot_{\mm^{(n)}}   (I_{\C^n}\otimes F)$, which is positive from $\mmn$ into $\elle{\kkn}$.

(3) Trivial.

(4) Since $A^{\frac 12} ,\, (I_\hh - A)^{\frac 12} \in \mmp$, by item (3) we have $A^{\frac 12} \odot_\mm A^{\frac 12} = A \odot_\mm I_\hh$ and $(I_\hh - A)^{\frac 12} \odot_\mm (I_\hh - A)^{\frac 12} = (I_\hh - A) \odot_\mm I_\hh = \ii_{\mm,\lh} - A \odot_\mm I_\hh$. Therefore,
$$
A^{\frac 12} \odot_\mm A^{\frac 12} = \ii_{\mm,\lh} - (I_\hh - A)^{\frac 12} \odot_\mm (I_\hh - A)^{\frac 12} \, ,
$$
and the claim follows as $(I_\hh - A)^{\frac 12} \odot_\mm (I_\hh - A)^{\frac 12} \succeq 0$.
\end{proof}

The importance of the elementary maps $E\odot_\mm F$'s will become clear in the following, as we will briefly see that by Kraus theorem (see Theorem \ref{Teo. Stines.} below) every map in $\cb{\mm,\lk}$ is the limit (in a suitable sense) of sums of elementary maps $E \odot_\mm F$.

Two of the main features of CB weak*-continuous maps which we will need in the rest of the paper are the following:
\begin{itemize}
\item[-] a notion of limit can be defined for a particular class of sequences in $\cb{\mm,\nn}$, which is the analogue of the least upper bound for increasing bounded sequences of operators;
\item[-] if $\mm_1$, $\mm_2$, $\nn_1$, $\nn_2$ are von Neumann algebras, the maps in $\cb{\mm_1,\nn_1}$ and $\cb{\mm_2,\nn_2}$ can be tensored in order to obtain elements of the set $\cb{\mm_1\votimes\mm_2,\nn_1\votimes\nn_2}$.
\end{itemize}
As these concepts are the main two ingredients in our definition of supermaps and in the proof of a dilation theorem for them, we devote the next two sections to their explanation.

\subsection{Increasing nets of normal CP maps}\label{sez. succ. cresc.}

If $\Lam$ is a directed set and $\{ A_\lam \}_{\lam\in\Lam}$ is a net of operators in $\mm_+$, we say that the net is {\em increasing} if $A_{\lam_1} \leq A_{\lam_2}$ whenever $\lam_1 \leq \lam_2$, and {\em bounded} if there exists $B\in\mm_+$ such that $A_{\lam} \leq B$ for all $\lam\in\Lam$. In this case, the net has a {\em least upper bound} $A\in\mm_+$, and we use the notation $A_\lam\uparrow A$.

We now extend the notion of increasing net and least upper bound to nets in $\cp{\mm,\nn}$. We say that the net $\left\{ \ee_\lam \right\}_{\lam\in\Lam}$ of elements in $\cp{\mm,\nn}$ is
\begin{itemize}
\item {\em CP-increasing} if $\ee_{\lam_1} \preceq \ee_{\lam_2}$ whenever $\lam_1 \leq \lam_2$,
\item {\em CP-bounded} if there exists a map $\ff\in\cp{\mm,\nn}$ such that $\ee_\lam \preceq \ff$ for all $\lam\in\Lam$.
\end{itemize}

Note that, if the net $\left\{ \ee_\lam \right\}_{\lam\in\Lam}$ is CP-increasing, then, for all $A\in\mm_+$, the net of operators $\left\{ \ee_\lam (A) \right\}_{\lam\in\Lam}$ is increasing in $\nn_+$. Moreover, if $\left\{ \ee_\lam \right\}_{\lam\in\Lam}$ is CP-bounded by $\ff\in\cp{\mm,\nn}$, then the net $\left\{ \ee_\lam (A) \right\}_{\lam\in\Lam}$ is bounded by $\ff(A)$ in $\nn$.

The following result now shows that the least upper bound exists for any CP-incresing and CP-bounded net in $\cp{\mm,\nn}$.
\begin{theorem}{Proposition}\label{Teo. Berb. 2}
If $\left\{ \ee_\lam \right\}_{\lam\in\Lam}$ is a net in $\cp{\mm,\nn}$ which is CP-increasing and CP-bounded, then there exists a unique $\ee\in\cp{\mm,\nn}$ such that
\begin{equation}\label{conv. di En}
\wklim_{\lam\in\Lam} \ee_\lam (A) = \ee(A) \quad \forall A\in \mm \, .
\end{equation}

$\ee$ has the following property: $\ee_\lam \preceq \ee$ for all $\lam\in\Lam$, and, if $\ff\in \cp{\mm,\nn}$ is such that $\ee_\lam \preceq \ff$ for all $\lam\in\Lam$, then $\ee \preceq \ff$.
\end{theorem}
\begin{proof}
We have just seen that, if $A\in\mm_+$, then the sequence $\left\{ \ee_\lam (A) \right\}_{\lam\in\Lam}$ is bounded and increasing in $\nn$. We thus define $\ee(A)\in\nn_+$ to be its least upper bound. Now, every operator in $\mm$ is the linear combination of four elements in $\mm_+$, therefore we can extend the definition of $\ee$ to all $A\in\mm$ by linearity (it is easy to see that such definition of $\ee(A)$ does not depend on the chosen decomposition of $A$ into positive operators). If $A=\sum_{k=0}^3 i^k A_k$, with $A_k\in\mm_+$, then $\ee_\lam (A_k) \uparrow \ee(A_k)$ for all $k$, hence Eq.~\eqref{conv. di En} follows by linearity.

In order to show that $\ee$ is normal, pick any positive sequence $\left\{ A_n \right\}_{n\in\N}$ in $\mm$ such that $A_n\uparrow A$ for some $A\in\mm_+$. Then, for all positive elements $\rho$ in the predual $\nn_\ast$ of $\nn$,
\begin{eqnarray*}
\rho(\ee(A)) & = & \sup_\lam \rho(\ee_\lam (A)) = \sup_\lam \sup_n \rho(\ee_\lam (A_n)) = \sup_n \sup_\lam \rho(\ee_\lam (A_n)) \\
& = & \sup_n \rho(\ee (A_n))
\end{eqnarray*}
Hence $\ee (A_n) \uparrow \ee (A)$, and $\ee$ is normal. 

Finally, to show that $\ee$ is CP, note that, for all $\tilde{A}\in\mmn_+$, we have $\wklim_{\lam\in\Lam} (\ii_n \otimes \ee_\lam) (\tilde{A}) = (\ii_n \otimes \ee) (\tilde{A})$ by Eq.~\eqref{conv. di En}. Since $(\ii_n \otimes \ee_\lam) (\tilde{A}) \geq 0$ for all $\lam$, it follows that $(\ii_n \otimes \ee) (\tilde{A}) \geq 0$.  Hence, $\ee$ is CP.

The remaining properties of $\ee$ are easy consequences of its definition and of the analogous properties of least upper bounds in $\mm$, $\nn$.
\end{proof}

If $\{\ee_\lam \}_{\lam\in\Lam}$ and $\ee$ are as in the statement of the above proposition, then we write $\ee_\lam \Uparrow \ee$.

We can now formulate Kraus theorem \cite{Kraus71} for normal CP maps in terms of CP-increasing and CP-bounded nets. To this aim, note that, if $I$ is any set, then the class of its finite subsets $\Lam_I$ is a directed set under inclusion.

\begin{theorem}{Theorem}\label{Teo. Stines.}
{\rm (Kraus theorem)}
Suppose $\mm\subset\lh$ is a von Neumann algebra. We have the following facts.
\begin{enumerate}
\item If $I$ is a finite or countable set and $\{E_i\}_{i\in I}$ are elements in $\elle{\kk, \hh}$ such that the net of partial sums  $\{\sum_{i\in J} E_i^\ast E_i\}_{J\in\Lam_I}$ is bounded in $\lk$, then the net of partial sums $\{\sum_{i\in J} E_i^\ast \odot_\mm E_i\}_{J\in\Lam_I}$ is CP-bounded and CP-increasing in $\cp{\mm,\lk}$, hence it converges in the sense of Proposition \ref{Teo. Berb. 2} to a unique limit $\ee\in\cp{\mm,\lk}$.
\item If $\ee\in\cp{\mm,\lk}$, then there exists a finite or countable set $I$ and a sequence $\{ E_i \}_{i\in I}$ of elements in $\elle{\kk, \hh}$ such that the net of partial sums $\{\sum_{i\in J} E_i^\ast \odot_\mm E_i\}_{J\in\Lam_I}$ converges to $\ee$ in $\cp{\mm,\lk}$ in the sense of Proposition \ref{Teo. Berb. 2}.
\end{enumerate}
In both cases,  choosing an arbitrary ordering $i_1,i_2,i_3 \ldots$ of the elements of $I$, we have that the sequence of partial sums $\{\sum_{k=1}^n E_{i_k}^\ast \odot_\mm E_{i_k}\}_{n\in\N}$ is CP-bounded and CP-increasing, and converges to $\ee$ in the sense of Proposition \ref{Teo. Berb. 2}.
\end{theorem}
\begin{proof}
(1) The claim is trivial when $\# I<\infty$, therefore we assume $I=\N$. If $J_1 , J_2 \in \Lam_\N$ with $J_1\leq J_2$, then
$$
\sum_{i\in J_2} E_i^\ast \odot_\mm E_i - \sum_{i\in J_1} E_i^\ast \odot_\mm E_i = \sum_{i\in J_2\setminus J_1} E_i^\ast \odot_\mm E_i \succeq 0 \, ,
$$
hence the net of partial sums is CP-increasing. To show that it is CP-bounded, we introduce the following bounded operator
$$
V : \kk \frecc \hh\otimes\ell^2 \, , \qquad Vv := \sum_{i\in \N} E_i v \otimes \delta_i \, ,
$$
where $\ell^2$ is the Hilbert space of square-summable sequences and $\{\delta_i\}_{i\in \N}$ is its standard basis. The sum converges in the norm topology of $\hh\otimes\ell^2$, as $\sum_{i\in \N} \no{E_i v}^2 \leq \scal{v}{Bv} < \infty$, where $B\in\lk$ is any positive operator such that $\sum_{i\in J} E_i^\ast E_i \leq B$ for all $J\in\Lam_\N$. Given $J\in \Lam_\N$, we let $P_J$ be the orthogonal projection of $\ell^2$ onto the linear span of $\{\delta_i \mid i\in J\}$. Moreover, we define the following normal $\ast$-homomorphism
$$
\pi : \mm\frecc \mm\votimes \C I_{\ell^2} \, , \qquad \pi(A) = A\otimes I_{\ell^2} \, ,
$$
and the map
$$
\ff: = (V^\ast \odot_{\mm\votimes \C I_{\ell^2}} V) \, \pi \, .
$$
As $\ff$ is the composition of normal CP maps, we have $\ff\in\cp{\mm,\lk}$. We claim that $\sum_{i\in J} E_i^\ast \odot_\mm E_i \preceq \ff$ for all $J\in\Lam_\N$. Indeed, we have
$$
\sum_{i\in J} E_i^\ast \odot_\mm E_i = (V^\ast \odot_{\elle{\hh\otimes\ell^2}} V) [(I_\hh \otimes P_J) \odot_{\mm\votimes \C I_{\ell^2}} (I_\hh \otimes P_J)] \, \pi \, ,
$$
and $(I_\hh \otimes P_J) \odot_{\mm\votimes \C I_{\ell^2}} (I_\hh \otimes P_J) \preceq \ii_{\mm\votimes \C I_{\ell^2},\, \elle{\hh\otimes\ell^2}}$ by item (4) of Proposition \ref{prop:prop. di odot}. Thus,
\begin{align*}
\sum_{i\in J} E_i^\ast \odot_\mm E_i & = (V^\ast \odot_{\elle{\hh\otimes\ell^2}} V) [(I_\hh \otimes P_J) \odot_{\mm\votimes \C I_{\ell^2}} (I_\hh \otimes P_J)] \, \pi\\
  &   \preceq (V^\ast \odot_{\mm\votimes \C I_{\ell^2}} V) \, \pi \\
   & = \ff ,
\end{align*}
and the claim follows.

(2) If $\ee\in\cp{\mm,\lk}$, then by Theorem 4.6 in \cite{EvLew} (or also Theorem 4.3 p.~165 in \cite{AJP}) there exists a finite or countable set $I$ and a sequence $\{ E_i \}_{i\in I}$ of elements in $\elle{\kk, \hh}$ such that
\begin{equation}\label{eq:Kraus}
\ee(A) = \sum_{i\in I} E_i^\ast A E_i \quad \forall A\in\mm \, ,
\end{equation}
where the series converges in the weak*-topology and is independent of the ordering of $I$. In particular, the net of partial sums $\{\sum_{i\in J} E_i^\ast E_i\}_{J\in\Lam_I}$ is bounded by $\ee(I_\hh)$ in $\lk$, hence by item (1) the net $\{\sum_{i\in J} E_i^\ast \odot_\mm E_i\}_{J\in\Lam_I}$ converges in the sense of Proposition \ref{Teo. Berb. 2} to a unique $\ee'\in\cp{\mm,\lk}$. Comparing Eqs.~\eqref{conv. di En} and \eqref{eq:Kraus}, we see that $\ee = \ee'$.

The last statement follows considering the subnet $\{\sum_{i\in J_n} E_i^\ast \odot_\mm E_i\}_{n\in\N}$ of the net $\{\sum_{i\in J} E_i^\ast \odot_\mm E_i\}_{J\in\Lam_I}$, where $J_n := \{i_1,i_2,\ldots ,i_n\}$, and by uniqueness of the limit.
\end{proof}

If $\ee$ and $\{ E_i \}_{i\in I}$ are as in item (2) of the above theorem, then we say that the expression $\sum_{i\in I} E_i^\ast \odot_\mm E_i$ is the {\em Kraus form} of $\ee$. Note that $\ee$ is a quantum channel (unital map) iff $I_\kk$ is the least upper bound of the net $\{\sum_{i\in J} E_i^\ast E_i\}_{J\in\Lam_I}$ in $\lk$.

Kraus theorem and Theorem \ref{CB = span CP} show that every map $\ee\in\cb{\mm,\lk}$ can be decomposed into a (possibly infinite) sum of elementary maps $E_i\odot_\mm F_i$. Indeed, by Theorem \ref{CB = span CP} we can choose four elements $\ee_k\in\cp{\mm,\lk}$ ($k=0,1,2,3$) such that $\ee = \sum_{k=0}^3 i^k \ee_k$, and each $\ee_k$ can be written in the Kraus form $\ee_k = \sum_{i\in I_k} E^{(k)\ast}_i \odot_\mm E^{(k)}_i$. It is clear, however, that such decomposition is not unique even if $\ee\in\cp{\mm,\lk}$ itself.

\begin{definition}{Remark}\label{rem:opspa1}
{\rm (The space $\cb{\mm,\lk}$ as a dual operator space)}
It is interesting to note that, if $\mm = \lh$, the linear space $\cb{\lh,\lk}$ is a {\em dual operator space} in the sense of operator space theory (see e.g.~1.2.20 in \cite{BlM} for the definition of dual operator spaces, and 1.2.19 in \cite{BlM} or Proposition 14.7 in \cite{Paul} for the operator space structure of $\cb{\mm,\lk}$). Indeed, this is proven in Proposition 2.1 of \cite{BS}. In the same reference, it also is proven that the operator space $\cb{\lh,\lk}$ is completely isometrically isomorphic to the {\em weak*-Haagerup tensor product} $\elle{\hh,\kk}\whotimes\elle{\kk,\hh}$ (see \cite{BS} or 1.6.9 in \cite{BlM} for the definition). Moreover, still in the case $\mm = \lh$, Kraus theorem \ref{Teo. Stines.} above is a restatement of Theorem 2.2 in \cite{BS}, which asserts that each $\ee\in\cb{\lh,\lk}$ is the weak*-limit of a sequence of maps $\{\sum_{k=1}^n E^\ast_k \odot_{\lh} F_k\}_{n\in\N}$ for some sequence of operators $\{E_k\}_{k\in\N}$ and $\{F_k\}_{k\in\N}$ in $\elle{\kk,\hh}$.  However, for simplicity of presentations in the  following we will not phrase our result in the language of dual operator spaces, because most of the proofs are simpler (and more intuitive) in the language of operator algebras.
\end{definition}

\subsection{Tensor product of weak*-continuous CB maps}\label{sez. amplificaz.}

If $\ee : \elle{\hh_1} \frecc \elle{\kk_1}$ and $\ff : \elle{\hh_2} \frecc \elle{\kk_2}$ are linear {\em bounded} maps, their tensor product $\ee\otimes\ff$ is well defined as a linear map $\elle{\hh_1} \hotimes \elle{\hh_2} \frecc \elle{\kk_1} \hotimes \elle{\kk_2}$. However, unless $\hh_1$ and $\kk_1$, or alternatively $\hh_2$ and $\kk_2$, are finite dimensional, in general one can not extend $\ee\otimes\ff$ to a map $\elle{\hh_1\otimes\hh_2} \frecc \elle{\kk_1\otimes\kk_2}$. Weak*-continuous CB maps constitute an important exception to this obstruction, as it is shown by the following proposition (see also Proposition 5.13 p.~228 in \cite{Tak}).

\begin{theorem}{Proposition}\label{tensor product of maps}
Let $\mm_1$, $\mm_2$, $\nn_1$, $\nn_2$ be von Neumann algebras. Given two  maps $\ee \in \cb{\mm_1 , \nn_1}$ and $\ff \in \cb{\mm_2 , \nn_2}$, there is a unique map $\ee \otimes \ff \in \cb{\mm_1 \votimes \mm_2 , \nn_1 \votimes \nn_2}$ such that
\begin{equation}\label{def on products}
(\ee \otimes \ff)  (A \otimes B)  =  \ee(A)  \otimes \ff (B)  \quad \forall A \in \mm_1 , \, B \in \mm_2 \, .
\end{equation}
If $\ee$ and $\ff$ are CP, then $\ee \otimes \ff \in \cp{\mm_1 \votimes \mm_2 , \nn_1 \votimes \nn_2}$.
\end{theorem}
\begin{proof}
Without loss of generality, let us assume $\mm_k\subset\elle{\hh_k}$ and $\nn_k\subset\elle{\kk_k}$ for $k=1,2$. First suppose that the maps $\ee$ and $\ff$ are CP, and have Kraus forms $\ee = \sum_{i\in I} E_i^\ast \odot_{\mm_1} E_i$ and $\ff  = \sum_{j \in  J}  F_j^\ast \odot_{\mm_2} F_j$. We can then define a map $\gg \in \cp{\mm_1 \votimes \mm_2 , \elle{\kk_1 \otimes \kk_2}}$, with Kraus form
$$
\gg  := \sum_{(i,j)\in I\times J}  (E_i \otimes F_j)^\ast  \odot_{\mm_1\votimes\mm_2} (E_i \otimes F_j) \, .
$$
It is easy to check that $\gg(A\otimes B) = \ee(A) \otimes\ff(B)$ for all $A \in \mm_1$, $B \in \mm_2$, hence $\gg$ extends the linear map $\ee \otimes \ff : \mm_1 \hotimes \mm_2 \frecc \nn_1 \hotimes \nn_2$ defined in Eq.~\eqref{def on products} to a weak*-continuous CP map from $\mm_1\votimes\mm_2$ into $\elle{\kk_1\otimes\kk_2}$. Such extension is unique by weak*-density of $\mm_1 \hotimes \mm_2$ in $\mm_1 \votimes \mm_2$. Moreover, since $\gg (\mm_1\hotimes\mm_2) \subset \nn_1\hotimes\nn_2$, we have $\gg (\mm_1\votimes\mm_2) \subset \nn_1\votimes\nn_2$. Hence, $\gg \in \cp{\mm_1 \votimes \mm_2 , \nn_1 \votimes \nn_2}$.

The claim of the theorem for generic elements $\ee \in \cb{\mm_1 , \nn_1}$ and $\ff \in \cb{\mm_2 , \nn_2}$ then follows by linearity and Theorem \ref{CB = span CP}.
\end{proof}

The map $\otimes : \cb{\mm_1 , \nn_1} \times\, \cb{\mm_2 , \nn_2} \frecc \cb{\mm_1\votimes\mm_2 , \nn_1\votimes\nn_2}$ defined in Proposition \ref{tensor product of maps} is clearly bilinear, and yelds the inclusion
$$
\cb{\mm_1 , \nn_1} \hotimes\, \cb{\mm_2 , \nn_2} \subset \cb{\mm_1\votimes\mm_2 , \nn_1\votimes\nn_2} \, .
$$

\begin{definition}{Remark}\label{rem:id dei CBn}
When $\mm_1 = M_m (\C)$ and $\nn_1 = M_n (\C)$, the product $\ee\otimes\ff$ defined in Proposition \ref{tensor product of maps} clearly coincides with the algebraic product that we already encountered in the definition of CB and CP maps (Definition \ref{def:CB-CP}). Moreover, the above inclusion actually becomes the equality
\begin{equation}\label{eq:eqoftens}
\cb{M_m (\C) , M_n (\C)} \hotimes\, \cb{\mm , \nn} = \cb{\mm^{(m)} , \nnn} \, .
\end{equation}
Indeed, choose two bases $\{f_i\}_{i=1}^{m^2}$ of $M_m (\C)$ and $\{g_j\}_{j=1}^{n^2}$ of $M_n (\C)$. For a map $\tilde{\ee}\in\cb{\mm^{(m)} , \nnn}$, define
$$
\tilde{\ee}_{ji} (A) := (g^\dag_j \otimes\ii_\nn) [\tilde{\ee} (f_i \otimes A)] \quad \forall A\in\mm
$$
(where the superscript $^\dag$ labels the dual basis). We then have $\tilde{\ee}_{ji} \in\cb{\mm , \nn}$, as $\tilde{\ee}_{ji}$ is obtained by composing and tensoring weak*-continuous CB maps (recall that the maps $g^\dag_j : M_n (\C) \frecc \C$ and $f_i : \C \frecc M_m (\C)$ are CB by Remark \ref{rem:CB=Lin in dim finita}). Since
$$
\tilde{\ee} = \sum_{i=1}^{m^2} \sum_{j=1}^{n^2} (g_j f^\dag_i) \otimes \tilde{\ee}_{ji} \, ,
$$
the equality of sets \eqref{eq:eqoftens} follows.
\end{definition}

It is easy to check that the tensor product $\otimes$ defined above preserves
\begin{itemize}
\item[-] composition of maps: $(\ee_1 \otimes \ff_1) (\ee_2 \otimes \ff_2) = \ee_1\ee_2  \otimes \ff_1 \ff_2$;
\item[-] ordering: if $\ee_1 \preceq \ee_2$ and $\ff_1 \preceq \ff_2$, then $\ee_1 \otimes \ff_1 \preceq \ee_2 \otimes \ff_2$;
\item[-] least upper bounds: if $\ee_\lam \Uparrow \ee$ and $\ff_\mu \Uparrow \ff $, then $\ee_\lam \otimes \ff_\mu \Uparrow  \ee\otimes \ff$ (where $(\lam_1 ,\mu_1) \leq (\lam_2 ,\mu_2)$ iff $\lam_1 \leq \lam_2$ and $\mu_1 \leq \mu_2$);
\item[-] quantum channels: if $\ee\in\cpn{\mm_1,\nn_1}$ and $\ff\in\cpn{\mm_2,\nn_2}$, then $\ee\otimes \ff \in \cpn{\mm_1 \votimes\mm_2 ,\nn_1 \votimes\nn_2}$.
\end{itemize}
Moreover, when tensoring the elementary maps $E_1\odot_{\mm_1} F_1$ and $E_2\odot_{\mm_2} F_2$, we clearly obtain
$$
(E_1\odot_{\mm_1} F_1) \otimes (E_2\odot_{\mm_2} F_2) = (E_1\otimes E_2)\odot_{\mm_1\votimes\mm_2} (F_1\otimes F_2) \, .
$$
In particular, we see that, if $\vv$ is another Hilbert space, then $(E\odot_\mm F)\otimes \ii_\vv = (E\otimes I_\vv)\odot_{\mm\votimes\elle{\vv}}(F\otimes I_\vv)$.

\section{Quantum supermaps}\label{sez. centr.}

In this section we introduce the central object in our study, i.e.~a particular set of linear maps $\SS : \cb{\mm_1,\nn_1} \frecc \cb{\mm_2,\nn_2}$ which mathematically describe the physically admissible transformations of quantum channels. These maps were introduced and studied in \cite{CDaP1,CDaP2} in the case where $\mm_i = \elle{\hh_i}$ and $\nn_i = \elle{\kk_i}$ are the full algebras of linear operators on finite dimensional Hilbert spaces $\hh_i$ and $\kk_i$. The main difference in the infinite dimensional case is the role of normality, which will be crucial for our dilation theorem (see Theorem \ref{teo. centr.} of the next section).

Let us start from some basic terminology:
\begin{definition}{Definition}
Suppose $\mm_1$, $\mm_2$, $\nn_1$, $\nn_2$ are von Neumann algebras. A linear map $\SS:  \cb{\mm_1,\nn_1} \frecc \cb{\mm_2,\nn_2}$ is
\begin{itemize}
\item[-] {\em positive} if $\SS (\ee) \succeq 0$ for all $\ee\succeq 0$;
\item[-] {\em completely positive (CP)} if the map 
$$\II_n \otimes \SS:  \cb{\mm_1^{(n)} , \nn_1^{(n)}} \to \cb{\mm_2^{(n)}, \nn_2^{(n)}}$$  is positive for every $n\in\N$, where $\II_n$ is the identity map on the linear space $\cb{M_n (\C),M_n (\C)}$;
\item[-] {\em normal} if $\SS(\ee_n) \Uparrow \SS(\ee)$ for all sequences $\{ \ee_n \}_{n\in\N}$ in $\cp{\mm_1,\nn_1}$ such that $\ee_n \Uparrow \ee$.
\end{itemize}
\end{definition}

Note that in the above definition of complete positivity we used the identification $\cb{\mmn , \nnn} = \cb{M_n (\C),M_n (\C)} \hotimes\, \cb{\mm , \nn}$ estabilished in Remark \ref{rem:id dei CBn}.

\begin{definition}{Remark}
Not every CP map $\SS:  \cb{\mm_1,\nn_1} \frecc \cb{\mm_2,\nn_2}$ is normal, even though, by definition, $\SS$ transforms normal maps into normal maps.  A simple example of non-normal CP map is the following: suppose $\mm_1 = \C$ and $\nn_1 = \lk$, with $\kk$ infinite dimensional. In this case, we have the natural identifications $\cb{\C,\lk} = \lk$ and $\cp{\C,\lk} = \lk_+$, and elements $\{\ee_n\}_{n\in\N}$ and $\ee$ in $\cp{\C,\lk}$ satisfy $\ee_n\Uparrow\ee$ iff $\ee_n (1)\uparrow\ee(1)$ in $\lk_+$. Consider a singular state $\rho :  \lk  \frecc \C$, i.e.~a positive functional such that $\rho (K)=0$ for every compact operator $K \in \lk$ and $\rho (I_\kk) =1$. Define the linear map $\SS: \cb{\C,\lk} \frecc \cb{\mm_2,\nn_2}$ given by $\SS (\ee)  =  \rho (\ee (1)) \ff$, where $\ff\in\cp{\mm_2,\nn_2}$ is fixed. Since $\rho$ is CP (see Proposition 3.8 in \cite{Paul}), it is easy to check that $\SS$ is CP. However, $\SS$ is not normal: consider for example a Hilbert basis $\{e_i\}_{i\in  \N}$ for $\kk$ and let $P_n$ be the orthogonal projection onto $\spanno{e_i \mid i \leq n}$. If $\ee_n,\ee\in\cp{\C,\lk}$ are given by $\ee_n (1) = P_n$ and $\ee (1) = I_\kk$, in this way one has $\ee_n \Uparrow  \ee$, whereas $\SS (\ee_n) =0$ and $\SS (\ee)  = \ff$. Hence, $\SS$ is not normal.     
\end{definition}

We are now in position to define quantum supermaps.
\begin{definition}{Definition}
A {\em quantum supermap} (or simply, \emph{supermap}) is a linear normal CP map $\SS : \cb{\mm_1,\nn_1} \frecc \cb{\mm_2,\nn_2}$.
\end{definition}

The convex set of quantum supermaps from $\cb{\mm_1,\nn_1}$ to $\cb{\mm_2,\nn_2}$ will be denoted by $\cpq{\mm_1 , \nn_1 ; \mm_2 , \nn_2}$. A partial order $\ll$ can be introduced in it as follows: given two maps $\SS_1, \SS_2 \in \cpq{\mm_1 , \nn_1 ; \mm_2 , \nn_2}$, we write $\SS_1 \ll \SS_2$ if  $\SS_2 - \SS_1 \in \cpq{\mm_1 , \nn_1 ; \mm_2 , \nn_2}$. 

We now specialize the definition of quantum supermaps to the following two main cases of interest.

\begin{definition}{Definition}
A quantum supermap $\SS \in \cpq{\mm_1 , \nn_1 ; \mm_2 , \nn_2}$ is
\begin{itemize}
\item[-] {\em deterministic} if it preserves the set of quantum channels, i.e.~if $\SS (\ee) \in \cpn{\mm_2, \nn_2}$ for all $\ee\in\cpn{\mm_1,\nn_1}$;
\item[-] {\em probabilistic} if a deterministic supermap $\mathsf T \in \cpq{\mm_1, \nn_1; \mm_2,\nn_2 }$ exists such that $\SS \ll \mathsf T$.
\end{itemize}
\end{definition}
Deterministic supermaps are a particular case of probabilistic supermaps. We will label by $\cpqn{\mm_1, \nn_1; \mm_2,\nn_2}$ the subset of deterministic supermaps in $\cpq{\mm_1, \nn_1; \mm_2,\nn_2}$.

Obviously, composing two quantum supermaps one still obtains a supermap: if $\SS_1  \in \cpq{\mm_1, \nn_1; \mm_2,\nn_2}$ and $\SS_2 \in \cpq{\mm_2, \nn_2; \mm_3,\nn_3}$, the composition map $\SS_2 \SS_1$ is an element in $\cpq{\mm_1, \nn_1; \mm_3,\nn_3}$. Similarly, the composition of two probabilistic [resp.~deterministic] supermaps  is a probabilistic [resp.~deterministic] supermap.

We now introduce two examples of supermaps which will play a very important role in the next section.

\begin{theorem}{Proposition}\label{compo}
{\rm (Concatenation)}
Given two maps $\aa \in \cp{\nn_1, \nn_2}$ and $\bb \in \cp{\mm_2, \mm_1}$, define the map
$$
\CC_{\aa,\bb} : \cb{\mm_1,\nn_1} \frecc \cb {\mm_2,\nn_2} \, , \qquad \CC_{\aa,\bb} (\ee) = \aa \ee \bb \, .
$$
Then $\CC_{\aa,\bb}\in\cpq{\mm_1, \nn_1; \mm_2, \nn_2}$. Moreover, if $\aa$ and $\bb$ are quantum channels, then $\CC_{\aa,\bb}$ is deterministic.
\end{theorem}
\begin{proof}
The map $\CC_{\aa,\bb}$ is normal: if $\ee_n  \Uparrow \ee$, then the sequence $\{\aa \ee_n\bb\}_{n \in \mathbb N}$ is CP-increasing and CP-bounded by $\aa \ee \bb$. Using Proposition \ref{Teo. Berb. 2}, we have $\wklim_n \aa \ee_n\bb (A) = \aa \ee \bb (A)$ for all $A\in\mm_2$, hence $\aa \ee_n\bb \Uparrow \aa \ee\bb$, i.e.~$\CC_{\aa,\bb}$ is normal.    
To prove complete positivity, note that for every map $\tilde \ee \in \cb{\mmn,\nnn}$ one has $(\II_n \otimes \CC_{\aa,\bb}) (\tilde \ee)  =  (\ii_n \otimes \aa)  \tilde \ee (\ii_n \otimes \bb)$. Therefore, if $\tilde \ee \succeq 0$, then also $(\II_n \otimes \CC_{\aa,\bb})(\tilde \ee)\succeq 0$, hence $\II_n \otimes \CC_{\aa,\bb}$ is positive and $\CC_{\aa,\bb}$ is CP. Finally, if $\aa$ and $\bb$ are quantum channels, then $\aa\ee\bb\in\cpn{\mm_2,\nn_2}$ for all $\ee\in\cpn{\mm_1,\nn_1}$, i.e.~$\CC_{\aa,\bb}$ is deterministic.
\end{proof}

\begin{theorem}{Proposition}\label{ampli}
{\rm (Amplification)}
Suppose $\vv$ is a Hilbert space, and define the amplification supermap
$$
\PI_\vv : \cb{\mm,\nn} \frecc \cb{\mm\votimes \lv , \nn\votimes \lv} \, , \qquad \PI_\vv (\ee) = \ee \otimes \ii_\vv \, ,
$$
where we recall that $\ii_\vv := \ii_{\lv}$ (cf.~Proposition \ref{tensor product of maps} for the definition of tensor product). Then the map $\PI_\vv$ is a deterministic supermap, that is, $\PI_\vv \in \cpqn{\mm,\nn ; \mm\votimes \lv, \nn\votimes \lv}$.
\end{theorem}
\begin{proof}
If  $\ee_n\Uparrow  \ee$, then the sequence $\{\ee_n \otimes \ii_\vv\}_{n \in \mathbb N}$ is CP-increasing and CP-bounded by $\ee \otimes \ii_\vv$, hence $\ee_n \otimes \ii_\vv \Uparrow \tilde{\aa}$ for some $\tilde{\aa}\in\cp{\mm\votimes\lv,\nn\votimes\lv}$ by Proposition \ref{Teo. Berb. 2}. We have
\begin{align*}
\tilde{\aa} (A\otimes B) & = \wklim_n (\ee_n \otimes \ii_\vv) (A\otimes B) = \wklim_n \ee_n (A) \otimes B = \ee (A) \otimes B \\
& = (\ee \otimes \ii_\vv) (A\otimes B)
\end{align*}
for all $A\in\mm$ and $B\in\lv$, which implies $\tilde{\aa} = \ee \otimes \ii_\vv$ by Proposition \ref{tensor product of maps}. Thus, $\ee_n\otimes\ii_\vv \Uparrow \ee\otimes\ii_\vv$, i.e.~$\Pi_\vv$ is normal. Clearly, if $\ee$ is unital, so is $\Pi_\vv (\ee) = \ee \otimes \ii_\vv$. To prove complete positivity, note that for every $\tilde \ee\in \cp{\mmn , \nnn}$ we have $ (\II_n \otimes  \Pi_\vv)   (\tilde \ee) = \tilde \ee \otimes \ii_\vv \succeq 0$, hence $\II_n \otimes  \Pi_\vv$ is positive and $\Pi_\vv$ is CP.
\end{proof}

The main result in the next section is that every deterministic supermap in the set $\cpqn{\mm_1,\elle{\kk_1} ; \mm_2,\elle{\kk_2}}$ is the composition of an amplification followed by a concatenation.

\section{Dilation of deterministic supermaps}\label{sez. stine}

This section contains the central result of our paper, namely the following dilation theorem for deterministic supermaps. 

\begin{theorem}{Theorem}\label{teo. centr.}
{\rm (Dilation of deterministic supermaps)}
Suppose $\mm_1$, $\mm_2$ are von Neumann algebras. A linear map $\SS : \cb{\mm_1 , \elle{\kk_1}} \frecc \cb{\mm_2 , \elle{\kk_2}}$ is a deterministic supermap if and only if there exists a triple $(\vv ,\, V ,\, \ff)$, where
\begin{itemize}
\item[-] $\vv$ is a separable Hilbert space
\item[-] $V:\kk_2\frecc \kk_1 \otimes \vv$ is an isometry
\item[-] $\ff$ is a quantum channel in $\cpn{\mm_2,\mm_1\votimes\lv}$
\end{itemize}
such that
\begin{equation}\label{eq. centr. 2}
[\SS (\ee)](A) = V^\ast \left[ (\ee\otimes \ii_{\vv}) \ff(A) \right] V \quad \forall \ee\in\cb{\mm_1 , \elle{\kk_1}} \, , \, \forall A\in\mm_2 \, .
\end{equation}
The triple $(\vv ,\, V ,\, \ff)$ can always be chosen in a way that
\begin{equation}\label{dens. in hhat1}
\vv = \spannochiuso{(u^\ast\otimes I_{\vv}) Vv \mid u\in\kk_1 \, , \, v \in\kk_2} \, .
\end{equation}
\end{theorem}

In Eq.~\eqref{dens. in hhat1}, the adjoint $u^\ast$ of $u\in\kk_1$ is the linear functional $u^\ast : w\mapsto \scal{u}{w}$ on $\kk_1$.

\begin{definition}{Definition}\label{def:min}
If a Hilbert space $\vv$, an isometry $V:\kk_2\frecc \kk_1 \otimes \vv$, and a quantum channel $\ff\in\cpn{\mm_2,\mm_1\votimes\lv}$ are such that Eq.~\eqref{eq. centr. 2} holds, then we say that the triple $(\vv ,\, V ,\, \ff)$ is a {\em dilation} of the supermap $\SS$.  
If also Eq.~\eqref{dens. in hhat1} holds, then we say that the dilation $(\vv ,\, V ,\, \ff)$ is \emph{minimal}.
\end{definition}

The importance of the minimality property is highlighted by the following fact.
\begin{theorem}{Proposition}\label{prop: minimality}
Let  $(\vv ,\, V ,\, \ff)$  and $(\vv' ,\, V' ,\, \ff')$ be two  dilations of the deterministic supermap $\SS\in\cpqn{\mm_1, \elle{\kk_1} ; \mm_2, \elle{\kk_2}}$. If $(\vv ,\, V ,\, \ff)$ is minimal,  then there exists a unique isometry $W : \vv \frecc \vv^\prime$ such that $V^\prime = (I_{\kk_1} \otimes W) V$. Moreover, the relation $\ff(A) = (I_{\mm_1} \otimes W^\ast) \ff'(A) (I_{\mm_1} \otimes W)$ holds for all $A\in\mm_2$.
\end{theorem}

The proofs of Theorem \ref{teo. centr.} and Proposition \ref{prop: minimality} will be given at the end of this section.   
\begin{definition}{Remark}
In Proposition \ref{prop: minimality}, if also the dilation $(\vv' ,\, V' ,\, \ff')$ is minimal, then the isometry $W$ is actually unitary. Indeed, let $W' : \vv' \frecc \vv$ be the isometry such that $V = (I_{\kk_1} \otimes W') V'$. We have $V = (I_{\kk_1} \otimes W') (I_{\kk_1} \otimes W) V = (I_{\kk_1} \otimes W'W) V$. Uniqueness then implies $W'W=I_{\vv}$, hence $W$ is unitary.
\end{definition}

\begin{definition}{Remark}
As claimed at the end of the previous section, Theorem \ref{teo. centr.} shows that every deterministic supermap $\SS\in\cpqn{\mm_1,\elle{\kk_1} ; \mm_2,\elle{\kk_2}}$ is the composition of an amplification followed by a concatenation. Indeed, setting $\aa = V^\ast\odot_{\elle{\kk_1\otimes\vv}} V$, we have $\aa\in\cpn{\elle{\kk_1\otimes\vv},\elle{\kk_2}}$, and Eq.~\eqref{eq. centr. 2} gives $\SS =\CC_{\aa,\ff} \, \Pi_\vv$.
\end{definition}

\begin{definition}{Remark}
It is useful to connect Theorem \ref{teo. centr.} with Eq.~\ref{eq. intro 1} and the previous results of \cite{CDaP1,CDaP2}. So, let us assume $\mm_1=\elle{\hh_1}$ and $\mm_2\subset\elle{\hh_2}$. We claim that a linear map $\SS:\cb{\elle{\hh_1},\elle{\kk_1}}\frecc\cb{\mm_2,\elle{\kk_2}}$ is a deterministic supermap if and only if there exist two separable Hilbert spaces $\vv,\uu$ and two isometries  $V:\kk_2\frecc \kk_1 \otimes \vv$, $U:\hh_1\otimes\vv\frecc \hh_2 \otimes \uu$ such that
\begin{equation}\label{eq. centr. old}
[\SS (\ee)](A) = V^\ast \left[ (\ee\otimes \ii_\vv) (U^\ast(A\otimes I_\uu)U) \right] V 
\end{equation}
for all $\ee\in\cb{\elle{\hh_1} , \elle{\kk_1}}$ and $A\in\mm_2$. Indeed, by Stinespring theorem (Theorem 4.3 p.~165 in \cite{AJP} and the discussion following it) every quantum channel $\ff\in\cpn{\mm_2,\elle{\hh_1}\votimes\lv}=\cpn{\mm_2,\elle{\hh_1\otimes\vv}}$ can be written as
$$
\ff(A)=U^\ast(A\otimes I_\uu)U \quad \forall A\in\mm_2
$$
for some separable Hilbert space $\uu$ and some isometry $U:\hh_1\otimes\vv\frecc \hh_2 \otimes \uu$. 
Eq.~\eqref{eq. centr. old} then follows by Eq.~\eqref{eq. centr. 2}, thus recovering the main result of \cite{CDaP1,CDaP2}.
%\begin{corollary}
%ma
%\end{corollary}
\end{definition}

\begin{definition}{Remark}\label{rem:opspa2}
Theorem \ref{teo. centr.} can be compared with an analogous result in the theory of operator spaces, namely the Christensen-Effros-Sinclair-Pisier (CSPS) theorem for maps $\varphi : \lh\hagotimes\lh \frecc \lk$ which are {\em completely bounded (CB)} in the sense of operator spaces; here, $\lh\hagotimes\lh$ is the algebraic tensor product $\lh\hotimes\lh$ endowed with the operator space structure given by the Haagerup tensor norm (see Chapter 17 in \cite{Paul} for a review of these topics). Indeed, one can show that, if a linear map $\SS : \cb{\mm_1,\elle{\kk_1}} \frecc \cb{\mm_2,\elle{\kk_2}}$ is CP {\em and} probabilistic, then it is automatically CB. In this case, if moreover $\mm_i = \elle{\hh_i}$, regarding the linear spaces $\cb{\elle{\hh_i},\elle{\kk_i}}$ as dual operator spaces according to Remark \ref{rem:opspa1}, normality of $\SS$ is equivalent to its weak*-continuity. These facts can be proven with some efforts as direct consequences of definitions, or more easily checked {\em a posteriori} by making use of Eq.~\eqref{eq. centr. 2} in Theorem \ref{teo. centr.}. Being an operator space, $\cb{\mm_2,\elle{\kk_2}}$ can be completely isometrically immersed into some $\lk$ by Ruan theorem (Theorem 13.4 in \cite{Paul}). On the other hand, by the completely isometric isomorphism $\cb{\elle{\hh_1},\elle{\kk_1}} \simeq \elle{\hh_1,\kk_1}\whotimes\elle{\kk_1,\hh_1}$ explained in Remark \ref{rem:opspa1}, $\SS$ can be regarded as a CB map from $\elle{\hh_1,\kk_1}\hagotimes\elle{\kk_1,\hh_1}$ into $\lk$. Assuming $\hh_1 = \kk_1 = \hh$, CSPS theorem (in the form of Theorem 17.12 of \cite{Paul}) then applies, implying the existence of an Hilbert space $\uhat$, two operators $S,T : \kk\frecc\uhat$ and two unital $\ast$-homomorphisms $\pi_1,\pi_2: \lh \frecc \elle{\uhat}$ such that
\begin{equation}\label{eq:CSPS}
\SS(E\otimes F) = S^\ast \pi_1 (E) \pi_2 (F) T \quad \forall E,F\in\lh \, .
\end{equation}

However, we stress that this expression is very different from the the dilation of Theorem \ref{teo. centr.} above for deterministic supermaps. In particular, our central Eq.~\eqref{eq. centr. 2} {\em does not} follow from Eq.~\eqref{eq:CSPS} in any way. The main novelty of Theorem \ref{teo. centr.} with respect to CSPS theorem may be traced back to the requirement that deterministic supermaps preserve quantum channels. Indeed, this is a very strong request, which can not be employed in the CSPS dilation of Eq.~\eqref{eq:CSPS} for the reason that Ruan theorem gives no means to characterize the image of the subset of quantum channels $\cpn{\mm_2,\elle{\kk_2}}$ under the immersion $\cb{\mm_2,\elle{\kk_2}} \hookrightarrow \lk$. In other words, it is not possible to translate the requrement that a deterministic supermap $\SS$ preserves the set of quantum channels into Eq.~\eqref{eq:CSPS}. Instead, we will see that, in order to prove Theorem \ref{teo. centr.}, one needs to explicitely construct {\em two} Stinespring-type dilations $(\uhat_1,\pi_1,U_1)$ and $(\uhat_2,\pi_2,U_2)$ associated to $\SS$ (see the proof of Proposition \ref{teo. centr. prel.} below), and make an essential use of the quantum channel preserving property in the construction of the dilation $(\uhat_1,\pi_1,U_1)$ (via Lemma \ref{lemma agg.} below).

Of course, one can recover our dilation \eqref{eq. centr. 2} from CSPS Eq.~\eqref{eq:CSPS} in the simple case $\mm_2 = \C$, for which the equality $\cb{\mm_2,\elle{\kk_2}} = \elle{\kk_2}$ is trivial and does not require Ruan theorem. We leave the details of the proof to the reader. Note however that even in this case the proof  still needs an application of Lemma \ref{lemma agg.} below.
\end{definition}

\begin{definition}{Remark}\label{rem:teo.centr.pred.}
As anticipated in the Introduction, Eq.~\eqref{eq. centr. 2}  shows that all deterministic supermaps can be obtained by connecting quantum devices in suitable circuits.   
   Such a physical interpretation is clear in the Schr\"odinger picture: indeed, turning Eq.~\eqref{eq. centr. 2} into its predual, we obtain  
$$
[\SS (\ee)]_\ast  (\rho) = \ff_\ast \left[(  \ee \otimes \ii_\vv)_\ast ( V \rho V^\ast)  \right]
$$
for all elements $\rho$ in the set $\trcl{\kk_2}$ of trace class operators on $\kk_2$ and $\ee \in \cb{\mm_1 , \elle{\kk_1}}$.  
The above equation means that the higher-order transformation $\SS$ can be obtained in the following way:  
\begin{enumerate}
\item apply an invertible transformation (corresponding to the isometry $V$), which transforms  the system $\kk_2$ into the composite system $\kk_1 \otimes \vv$;
\item use the input device (corresponding to the predual quantum operation $\ee_*$) on system $\kk_1$, thus transforming it into system $\hh_1$;
\item apply a physical transformation (corresponding to the predual channel $\ff_*$). 
\end{enumerate}
In particular, it $\mm_i=\elle{\hh_i}$, we can take the Stinespring dilation $\ff(A) = U^\ast (A\otimes I_\uu) U$ of $\ff$. The last equation then rewrites
$$
[\SS (\ee)]_\ast  (\rho) = {\rm tr}_\uu \left\{ U \left[( \ee \otimes \ii_\vv)_\ast ( V \rho V^\ast)  \right] U^\ast \right\}
$$
where ${\rm tr}_\uu$ denotes the partial trace over $\uu$. If $\rho$ is a quantum state (i.e.~$\rho\geq 0$ and $\trt{\rho} =1$), this means that the quantum system with Hilbert space $\kk_2$ first undergoes the invertible evolution $V$, then the quantum channel $(\ee \otimes \ii_\vv)_\ast$, and finally the invertible evolution $U$, after which the ancillary system with Hilbert space $\uu$ is discarded.  It is interesting to note that the same kind of  sequential composition of  invertible evolutions  also appears in a very different context: the reconstruction  of quantum stochastic processes from correlation kernels \cite{belavkin,lindblad,parthasarathy}. That context  is very different from the present framework of higher-order maps, and it is a remarkable feature of Theorem \ref{teo. centr.}  that any deterministic supermap on the space of quantum operations can be achieved through a two-step sequence of invertible evolutions.  
\end{definition}

Theorem \ref{teo. centr.} contains as a special case the Stinespring dilation of quantum channels. This fact is illustrated in the following two examples.

\begin{definition}{Example}
Suppose that  $\mm_1  = \mm_2 = \C$, the trivial von Neumann algebra. In this case we have the identification $\cb{\C , \elle{\kk_i}} = \elle{\kk_i}$. Precisely, the element $\ee\in \cb{\C , \kk_i}$ is identified with the operator $A_\ee = \ee (1) \in \elle{\kk_i}$. Using the fact that  $\cpn{\mm_2,\mm_1\votimes\lv} = \{ I_\vv \}$ we then obtain that  Eq.~\eqref{eq. centr. 2} becomes
$$
[\SS(\ee)](1) = V^\ast (A_\ee\otimes I_\vv) V \, ,
$$
which is just Stinespring dilation for normal CP maps. A linear map $\SS : \elle{\kk_1} \frecc \elle{\kk_2}$ is thus in $\cpqn{\C,\elle{\kk_1} ; \C,\elle{\kk_2}}$ if and only if it is a unital normal CP map, i.e.~a quantum channel.
\end{definition}

\begin{definition}{Example}
Suppose now that $\kk_1 = \kk_2  =\C$. In this case we have the identification  $\cb{\mm_i , \C} = \mm_{i\, \ast}$, the predual space of $\mm_i$ (see e.g.~Proposition 3.8 in \cite{Paul}). Precisely, the CP map $\ee\in \cb{\mm_i , \C}$ is identified with the element $\rho_\ee\in\mm_{i\, \ast}$ given by $\ee(A) = \rho_\ee(A) \ \forall A\in\mm_i$.
Moreover, the isometry $V : \C  \to \C \otimes \vv = \vv$ is identified with a vector $v \in \vv$ with $\no {v} = 1$, and Eq.~\eqref{eq. centr. 2} becomes
\begin{align*}
[\SS (\ee)](A) & =  \scal{v}{\left[(\ee\otimes \ii_{\vv}) \ff(A)\right] v}\\
& = (\rho_\ee\otimes \omega_v)(\ff(A)) \\
& = [\ff_\ast (\rho_\ee\otimes \omega_v)] (A) \, ,
\end{align*}
where $\omega_v \in \elle{\vv}_\ast$ is the linear form $\omega_v : A\mapsto \scal{v}{Av}$. Thus, $\SS (\ee) = \ff_\ast (\rho_\ee\otimes \omega_v)$, hence $\SS$, viewed as a linear map $\SS: \mm_{1\, \ast} \frecc \mm_{2\, \ast}$, is CP and trace preserving. In other words, $\SS$ is a quantum channel in the Schr\"odinger picture.
\end{definition}

The rest of this section is devoted to the  proof of Theorem \ref{teo. centr.}, which first requires some auxiliary lemmas.

\begin{theorem}{Lemma}\label{to lemma agg.}
Suppose $\mm \subset \lh$, and let $\SS \in \cpqn{\mm,\lh;\nn,\lk}$. If $\ee, \ff \in \cp{\mm,\lh}$ are such that $\ee(I_\hh)  = \ff (I_\hh)$, then $[\SS(\ee)] (I_\nn)  =  [\SS(\ff)] (I_\nn)$. 
\end{theorem}
\begin{proof}
By linearity, it is enough to prove the claim for $\ee (I_\hh) = \ff (I_\hh) \leq I_\hh$.  Let $A := I_\hh - \ee (I_\hh)$,  $\aa:= A^{\frac 12} \odot_\mm A^{\frac 12}$, $\ee' :=\ee + \aa$, and $\ff' := \ff+\aa$. With this definition, $\ee^\prime ,\ff^\prime \in \cpn{\mm,\lh}$. Since $\SS$ is deterministic, one has 
$$
\begin{array}{lll}
I_\kk &= [\SS (\ee')] (I_\nn)  &= [\SS(\ee)] (I_\nn) + [\SS(\aa)](I_\nn) \\
I_\kk &= [\SS (\ff')] (I_\nn)  &= [\SS (\ff)](I_\nn) + [\SS(\aa)](I_\nn).
\end{array}
$$
By comparison, this implies that $[\SS(\ee)] (I_\nn) = [\SS (\ff)] (I_\nn)$.
\end{proof}

\begin{theorem}{Lemma}\label{lemma norm.}
Suppose $\mm \subset \lh$, and let $\SS \in \cpqn{\mm,\lh;\nn,\lk}$. Then, for all $\ee\in\cp{\lh,\lh}$,
$$
[\SS (\ee\ff)] (I_\nn) = [\SS (\left.\ee\right|_\mm)] (I_\nn) \quad \forall \ff\in\cpn{\mm,\lh} \, .
$$
\end{theorem}
\begin{proof}
Note that the restriction $\left.\ee\right|_\mm$ belongs to $\cp{\mm,\elle{\hh}}$ by Remark \ref{restr-CP}. Therefore, since $\ee\ff (I_\hh)  = \left.\ee\right|_\mm (I_\hh)$ for all $\ff\in\cpn{\mm,\lh}$, the claim is an immediate consequence of Lemma \ref{to lemma agg.}.
\end{proof}

\begin{theorem}{Lemma}\label{lemma agg.}
Suppose $\mm \subset \lh$, and let $\SS \in \cpqn{\mm,\lh;\nn,\lk}$. Then
$$
[\SS (\ee (I_\hh \odot_\mm A))] (I_\nn) = [\SS (\ee  (A\odot_\mm I_\hh))](I_\nn)
$$
for all $\ee \in\cb{\lh,\lh}$ and $A \in\lh$.
In particular, 
$$
[\SS (E \odot_\mm A F)] (I_\nn)  = [\SS ( EA  \odot_\mm  F)] (I_\nn)  \quad \forall E, F , A \in \lh \, .
$$
\end{theorem}
\begin{proof}
By linearity, it is enough to prove the claim for $A^* = A$ and  for $\ee \in \cp{\lh,\lh}$. One has
$$
A\odot_\mm I_\hh -  I_\hh \odot_\mm A   =  \frac 1{2i}  (\ee_+ - \ee_-) \, ,
$$
where $\ee_+ , \ee_- \in \cp{\mm,\lh}$ are given by
$$
\ee_\pm := (A \pm iI_\hh)^\ast \odot_\mm (A \pm iI_\hh) \, .
$$
Since $\ee_+(I_\hh) = \ee_-(I_\hh)$, we can apply Lemma \ref{to lemma agg.} to the maps $\ee\ee_+$ and $\ee\ee_-$, and obtain
\begin{align*}
& [\SS ( \ee (A \odot_\mm I_\hh))] (I_\nn)- [\SS(\ee (I_\hh \odot_\mm A))] (I_\nn) = \\
& \qquad \qquad \qquad \qquad \qquad \qquad = \frac{1}{2i} ([\SS (\ee\ee_+)] (I_\nn) - [\SS (\ee\ee_-)] (I_\nn)) \\
& \qquad \qquad \qquad \qquad \qquad \qquad = 0 \, ,
\end{align*}
hence the claim.

The last statement trivially follows taking $\ee=E\odot_\mm F$.
\end{proof}

\begin{theorem}{Lemma}\label{prop. sulla forma ass.}
Suppose $\mm_1\subset\elle{\hh_1}$, and let $\SS$ be a (not necessarily deterministic) supermap in the set $\cpq{\mm_1,\elle{\kk_1};\mm_2,\elle{\kk_2}}$. Define a sesquilinear form $\scal{\cdot}{\cdot}_\SS$ on the algebraic tensor product $\elle{\kk_1,\hh_1} \hotimes \mm_2 \hotimes \kk_2$ as follows
$$
\scal{E_1\otimes A_1 \otimes v_1}{E_2\otimes A_2 \otimes v_2}_\SS := \scal{v_1}{\left[ \SS \left( E_1^\ast \odot_{\mm_1} E_2 \right) \right] \left( A_1^\ast A_2 \right) v_2} \, .
$$
Then, the sesquilinear form $\scal{\cdot}{\cdot}_\SS$ is positive semidefinite.

If also $\TT\in \cpq{\mm_1,\elle{\kk_1};\mm_2,\elle{\kk_2}}$ and $\TT\ll\SS$, then
$$
0\leq \scal{\phi}{\phi}_\TT \leq \scal{\phi}{\phi}_\SS \quad \forall \phi\in \elle{\kk_1,\hh_1} \hotimes \mm_2 \hotimes \kk_2 \, .
$$
\end{theorem}
\begin{proof}
Let $\phi = \sum_{i=1}^n E_i \otimes A_i \otimes v_i$ be a generic element in the linear space $\elle{\kk_1,\hh_1} \hotimes \mm_2 \hotimes \kk_2$. 
Let $\{e_i\}_{i=1}^n$ be the standard basis for the Hilbert space $\C^n$, and $\{e_{ij}\}_{i,j=1}^n$ be the standard basis of the matrix space $M_n(\C)$, given by $e_{ij} (e_k) = \delta_{jk} e_i$. Define 
\begin{align*}
\tilde v &:=  \sum_{i=1}^n   e_i \otimes v_i  \in \kk_2^{(n)} \\
\tilde A &:=  \sum_{i=1}^n  e_{1i} \otimes A_i \in \mm_2^{(n)} \\
\tilde E &: = \sum_{i=1}^n e_{ii} \otimes E_i \in \elle{\kk_1^{(n)} , \hh_1^{(n)}} \, .
\end{align*}
With these definitions, we have $\tilde{E}^\ast \odot_{\mm_1^{(n)}} \tilde{E}  = \sum_{i,j=1}^n  (e_{ii} \odot_{M_n(\C)} e_{jj} ) \otimes (E_i^\ast \odot_{\mm_1} E_j)$ and $\tilde A^*\tilde A  =\sum_{i,j=1}^n e_{ij} \otimes A_i^* A_j$. Hence, we obtain
\begin{align*}
(\II_n \otimes \SS)  (\tilde{E}^\ast \odot_{\mm_1^{(n)}} \tilde E) & = \sum_{i,j} (e_{ii} \odot_{M_n(\C)} e_{jj} ) \otimes \SS (E_i^\ast \odot_{\mm_1} E_j) 
\end{align*}
and 
\begin{align*}
[(\II_n \otimes \SS)  (\tilde{E}^\ast \odot_{\mm_1^{(n)}} \tilde E) ] (\tilde A^* \tilde A ) &= \sum_{i,j} e_{ij}  \otimes [\SS (E_i^\ast \odot_{\mm_1} E_j)] (A^*_i A_j) . 
\end{align*}
Complete positivity of $\SS$ then implies
\begin{align*}
0  &\leq  \scal{\tilde v}{  [(\II_n \otimes \SS)(\tilde{E}^\ast \odot_{\mm_1^{(n)}} \tilde E)] (\tilde A^*\tilde A) \tilde v }    \\
&=  \sum_{i,j}   \scal{v_i}{[\SS(E_i^\ast \odot_{\mm_1} E_j)] (A_i^* A_j) v_j}   \\
& =  \scal{\phi}{\phi}_\SS \, ,
\end{align*}
which shows that the sesquilinear form $\scal{\cdot}{\cdot}_\SS$ is positive semidefinite.

Since the sesquilinear forms $\scal{\cdot}{\cdot}_\TT$, $\scal{\cdot}{\cdot}_\SS$ and $\scal{\cdot}{\cdot}_{\SS-\TT}$ are all positive semidefinite, the second statement in the lemma follows from
$$
\scal{\phi}{\phi}_\TT = \scal{\phi}{\phi}_\SS - \scal{\phi}{\phi}_{\SS-\TT} \leq \scal{\phi}{\phi}_\SS \, .
$$
\end{proof}

In the next two lemmas, we \emph{do not} assume separability as a part in the definition of Hilbert spaces.

\begin{theorem}{Lemma}\label{normal pi}
Let $\hh$ be separable, $\{e_i\}_{i\in \mathbb N}$ be a Hilbert basis for  $\hh$, and $P_n$ be the orthogonal projection onto $\spanno{  e_i \mid i \leq n}$. A unital $\ast$-homomorphism $\pi: \lh \frecc \elle{\uhat}$ (with $\uhat$ not assumed separable) is normal if and only if $\pi (P_n) \uparrow  I_{\uhat}$.
\end{theorem}
\begin{proof}
Since $P_n  \uparrow I_\hh$, if $\pi$ is normal one must necessarily have $\pi(P_n) \uparrow \pi (I_\hh) = I_\uhat$. Conversely, assume that $\pi (P_n) \uparrow I_\uhat$.   Let us decompose $\pi$ into the orthogonal sum of $\ast$-homomorphisms $\pi = \pi_{\rm nor} \oplus \pi_{\rm sin}$, where $\pi_{\rm nor}$ is normal and  $\pi_{\rm sin}$ is singular, that is  $\pi_{\rm sin} (K) = 0$ for every compact operator $K\in\lh$ (see e.g.~Proposition 10.4.13, p.~757 of \cite{KadRin}).  We then have $\pi (P_n) = \pi_{\rm nor}  (P_n) \uparrow \pi_{\rm nor}  (I_\hh)$ by normality, hence $\pi_{\rm nor}  (I_\hh) = I_\uhat$. On the other hand, $I_\uhat = \pi_{\rm nor} (I_\hh) \oplus \pi_{sin} (I_\hh)$. This implies $\pi_{\rm sin} (I_\hh) = 0$, and, therefore, $\pi_{\rm sin} =0$.
\end{proof}

\begin{theorem}{Lemma}\label{lemma sep. di K}
Let $\hh$ be separable and  $\pi : \lh \frecc \elle{\uhat}$ be a normal unital $\ast$-homomorphism (with $\uhat$ not assumed separable). If there exists a separable subset $\ss\subset\uhat$ such that the linear space
\begin{equation}\label{denso1}
\spanno{\pi (A) v \mid A\in\lh , v\in\ss}
\end{equation}
is dense in $\uhat$, then $\uhat$ is separable.
\end{theorem}
\begin{proof}
Since the Hilbert space $\hh$ is separable, the Banach subspace $\lzh$ of the compact operators in $\lh$ is separable. Let $P_n$ be defined as in the previous proposition. By normality of $\pi$, we have $\lim_n \no{\pi(P_n) v - v} = 0$ for all $v\in\uhat$ (Lemma 5.1.4 in \cite{KadRinI}). Therefore, $\pi (A) v = \lim_n \pi (AP_n)v$ for all $A\in \lh$ and $v\in\uhat$, where $AP_n\in\lzh$ because $P_n$ has finite rank. Therefore, the closure of the linear space defined in \eqref{denso1} coincides with the closure of the linear space spanned by the set $\left\{\pi (A) v \mid A\in\lzh , v\in\ss\right\}$, which is separable by separability of $\lzh$ and $\ss$ and by continuity of the mapping $\lzh \times \ss \ni (A,v) \mapsto \pi(A) v \in \uhat$. Separability of $\uhat$ then follows.
\end{proof}

We are now in position to prove the existence of the dilation of Theorem \ref{teo. centr.} in the special case $\mm_1 \subset \elle{\hh} =  \mm_2$ and $\dim \hh = \infty$.

\begin{theorem}{Proposition}\label{teo. centr. prel.}
Let $\dim \hh = \infty$, and assume $\mm\subset\lh$. If the linear map $\SS : \cb{\mm,\lh} \frecc \cb{\nn,\lk}$ is a deterministic supermap, then there exists a Hilbert space $\uu$, an isometry $U:\kk\frecc \hh \otimes \uu$ and a quantum channel $\gg\in\cpn{\nn,\mm\votimes\lu}$ such that
\begin{equation}\label{eq. centr.}
[\SS (\ee)](A) = U^\ast \left[(\ee\otimes \ii_\uu) \gg (A) \right] U \quad \forall \ee\in\cb{\mm,\lh} \, , \, \forall A\in\nn \, .
\end{equation}
\end{theorem}
\begin{proof}
Suppose that  $\SS : \cb{\mm,\lh} \frecc \cb{\nn,\lk}$ is a deterministic supermap. Let $\scal{\cdot}{\cdot}_1$ be the positive semidefinite sesquilinear form in $\lh \hotimes \kk$ defined by
$$
\scal{E_1\otimes v_1}{E_2\otimes v_2}_1: = \scal{E_1 \otimes I_\nn \otimes v_1}{ E_2 \otimes I_\nn \otimes v_2 }_\SS .
$$
Let $\rr_1$ be its kernel and $\uhat_1$ be the Hilbert space completion of the quotient space $\lh \hotimes \kk / \rr_1$ (not assumed separable). We denote by $\scal{\cdot}{\cdot}_1$ and $\no{\cdot}_1$ the scalar product and norm in $\uhat_1$.

Moreover, let $\rr_2$ be the kernel of the positive semidefinite sesquilinear form $\scal{\cdot}{\cdot}_\SS$, and let $\uhat_2$ be the Hilbert space completion (not assumed separable) of the quotient space $\lh \hotimes \nn \hotimes \kk / \rr_2$ with respect to such form. We denote by $\scal{\cdot}{\cdot}_2$ and $\no{\cdot}_2$ the resulting scalar product and norm in $\uhat_2$.

We define two linear maps
$$
\begin{array}{ll}
U_1: \kk \frecc \lh \hotimes \kk  \qquad  &U_1 v = I_\hh \otimes v \\
U_2 : \lh \hotimes \kk \frecc \lh \hotimes \nn \hotimes \kk  \qquad & U_2 (E\otimes v) = E\otimes I_\nn \otimes v \, .
\end{array}
$$
It is easy to verify that $U_1$ and $U_2$ extend to isometries $U_1: \kk \frecc \uhat_1$ and  $U_2: \uhat_1 \frecc \uhat_2$, respectively. Indeed, for $U_1$ we have the equality
\begin{align*}
\no{U_1 v }_1^2 &= \scal{I_\hh \otimes I_\nn \otimes v}{I_\hh \otimes I_\nn \otimes v}_\SS \\
&  = \scal{v}{[\SS(I_\hh \odot_\mm I_\hh)](I_\nn) v} \\
&= \no{v}^2 \, ,
\end{align*}
where we used the fact that $\SS$ is deterministic, implying $[\SS (I_\hh \odot_\mm I_\hh)] (I_\nn) = I_\kk$.
For $U_2$, taking $\phi = \sum_{i=1}^n  E_i \otimes v_i$ we have the equality
\begin{align*}
\no{U_2  \phi}_2^2  & = \sum_{i,j}  \scal{ E_i \otimes  I_\nn \otimes v_i} { E_j \otimes I_\nn \otimes v_j    }_\SS \\
& = \sum_{i,j}  \scal{ E_i \otimes v_i}{E_j \otimes v_j}_1  \\
& =  \no{\phi}_1^2 \, .  
\end{align*}

For $B\in\lh$, we introduce the linear operator $\pi_1 (B)$ on $\lh \hotimes \kk$ defined by
$$
[\pi_1 (B)] (E\otimes v) := B E\otimes v
$$
for all $E\in\lh$, $v\in\kk$. We claim that $\pi_1 (B)$ extends to a bounded linear operator on $\uhat_1$, that $\pi_1$ is a normal unital $\ast$-homomorphism of $\lh$ into $\elle{\uhat_1}$, and that $\uhat_1$ is separable. Indeed, for every  $\phi = \sum_{i=1}^n E_i \otimes v_i$ and $\psi = \sum_{j=1}^m F_j \otimes w_j$, we have
\begin{align*}
\scal{\phi}{\pi_1 (B) \psi }_1 & =  \sum_{i,j} \scal{v_i}{\left[ \SS (E_i^\ast \odot_\mm BF_j  ) \right] (I_\kk) w_j} \\
& =  \sum_{i,j} \scal{v_i}{\left[ \SS (E_i^\ast B \odot_\mm F_j  ) \right] (I_\kk) w_j} \\
& =  \scal{\pi_1 (B^\ast)\phi}{\psi}_1 ,
\end{align*}
where we used Lemma \ref{lemma agg.}. Note that $\pi_1 (I_\hh)$ is the identity on $\lh\hotimes\kk$, and
$$
\pi_1 (B_1) \pi_1 (B_2) = \pi_1 (B_1 B_2) \quad \forall B_1 ,B_2 \in \lh \, .
$$
It follows that, for all $\phi \in \lh \hotimes \kk$, the map $\omega_\phi : B \mapsto \scal{\phi}{\pi_1 (B) \phi}_1$ is a positive linear functional on $\lh$, hence
$$
\no{\pi_1 (B) \phi}_1^2= \omega_\phi (B^*B) \leq \no{B^*B}_\infty \omega_\phi (I_\hh) =\no{B}_\infty^2\no{\phi}_1^2 .
$$
Therefore,  $\pi_1 (B)$ extends to a bounded operator on $\uhat_1$, and $\pi_1$ is a unital $\ast$-homomorphism of $\lh$ into $\elle{\uhat_1}$.  
We now prove that $\pi_1$ is normal.  Let $\{e_i\}_{i\in \mathbb N} $ be a Hilbert basis for $\hh$, $Q_i$ be the orthogonal projection onto $\C e_i$, and $P_n$ be the orthogonal projection onto  $\spanno{e_i \mid  i\leq n }$. By Lemma \ref{normal pi}, to prove that $\pi_1$ is normal it is enough to prove that $\pi_1 (P_n)\uparrow I_{\uhat_1}$. For every $\phi = E\otimes v$, $\psi = F \otimes w$ we have
\begin{align*}
\scal{\phi}{\pi_1 (P_n) \psi}_1 & =  \sum_{i=1}^n \scal{\pi_1 (Q_i) \phi}{\pi_1 (Q_i) \psi} _1\\
&= \sum_{i=1}^n \scal{v}{[\SS (E^\ast Q_i \odot_\mm Q_i F)](I_\nn) w} \\
& =  \scal{v}{[\SS ((E^\ast \odot_{\lh} F) \ff_n )](I_\nn) w} ,
\end{align*}
where $\ff_n = \sum_{i=1}^n Q_i \odot_\mm Q_i \in \cp{\mm,\lh}$.
Let $\ff \in \cpn{\mm,\lh}$ be the quantum channel defined by $\ff_n \Uparrow \ff$. 
Using the polarization identity $E^\ast \odot_{\lh} F  = \frac{1}{4}  \sum_{k=0}^3 i^k  (i^k E + F)^\ast \odot_{\lh} (i^k E + F)$, the normality of $\SS$ and Lemma \ref{lemma norm.}, we then obtain
\begin{align*}
& \lim_n \scal{\phi}{\pi_1 (P_n) \psi}_1 
=\lim_n  \scal{v}{[\SS ((E^\ast \odot_{\lh} F) \ff_n )](I_\nn) w} \\
& \qquad \qquad \quad =  \frac 1 4 \sum_{k=0}^3    i^k    \lim_n  \scal{v}{[\SS (((i^k E + F)^\ast \odot_{\lh} (i^k E + F)) \ff_n )](I_\nn) w} \\
& \qquad \qquad \quad =  \frac 1 4 \sum_{k=0}^3    i^k      \scal{v}{[\SS (((i^k E + F)^\ast \odot_{\lh} (i^k E + F)) \ff )](I_\nn) w}\\
& \qquad \qquad \quad =  \frac 1 4 \sum_{k=0}^3    i^k      \scal{v}{[\SS ((i^k E + F)^\ast \odot_\mm (i^k E + F))](I_\nn) w}\\
& \qquad \qquad \quad = \scal{v}{[\SS (E^\ast \odot_\mm F)](I_\nn) w} \\
& \qquad \qquad \quad =  \scal{\phi}{\psi}_1 \, .
\end{align*}
This relation extends by linearity to all $\phi ,\psi  \in \lh\hotimes \kk$, and, since the sequence $\{ \pi_1 (P_n) \}_{n\in\N}$ is norm bounded, by density to all $\phi ,\psi \in \uhat_1$. Therefore, we obtain $\wklim_n  \pi_1 (P_n) = I_{\uhat_1}$, thus concluding the proof of normality of $\pi_1$. Note that the linear space $\spanno{E\otimes v = \pi_1 (E) U_1 v \mid E\in\lh,\, v\in\kk}$ is dense in $\uhat_1$ by definition, hence, using Lemma \ref{lemma sep. di K} with $\uhat \equiv \uhat_1$ and $\ss \equiv U_1\kk$,  we  obtain that $\uhat_1$ is separable.

For $C\in\nn$, we define a linear operator $\pi_2 (C)$ on $\lh \hotimes \nn \hotimes \kk$ given by
$$
[\pi_2 (C)] (E\otimes A\otimes v) := E\otimes C A\otimes v
$$
for all $E\in\lh$, $A\in\nn$, $v\in\kk$. We claim that $\pi_2 (C)$ extends to a bounded linear operator on $\uhat_2$ and that $\pi_2$ is a unital $\ast$-homomorphism of $\nn$ into $\elle{\uhat_2}$.
Indeed, for all vectors $\phi,\psi\in\lh\hotimes\nn\hotimes\kk$, with $\phi = \sum_{i=1}^n E_i \otimes A_i \otimes v_i$ and $\psi = \sum_{j=1}^m F_j \otimes B_j \otimes w_j$, we have
\begin{eqnarray*}
\scal{\phi}{\pi_2 (C) \psi}_2 & = & \sum_{i,j} \scal{v_i}{\left[ \SS (E_i^\ast \odot_\mm F_j ) \right] (A_i^\ast C B_j) w_j} \\
& = & \scal{\pi_2 (C^\ast)\phi}{\psi}_2 .
\end{eqnarray*}
Clearly, $\pi_2 (I_\nn)$ is the identity on $\lh\hotimes\nn\hotimes\kk$. Moreover, $\pi_2 (C_1) \pi_2 (C_2) = \pi_2 (C_1 C_2)$. 
The same argument used for $\pi_1$ then shows that $\pi_2 (C)$ extends to a bounded operator on $\elle{\uhat_2}$, and $\pi_2$ is a unital $\ast$-homomorphism of $\nn$ into $\elle{\uhat_2}$. 

We now introduce the following linear map
$$
\gg : \nn \frecc \elle{\uhat_1} \, , \qquad \gg(A) = U_2^\ast \pi_2 (A) U_2 \, .
$$
Clearly, the map $\gg$ is CP and unital. If $\{A_n\}_{n\in\N}$ is an increasing sequence in $\nn$ such that $A_n \uparrow A$, then, for all vectors $\phi,\psi\in\lh\hotimes\kk$, with $\phi = \sum_{i=1}^m E_i \otimes v_i$, $\psi = \sum_{j=1}^k F_j \otimes w_j$, we have
\begin{eqnarray*}
\lim_n \scal{\phi}{\gg(A_n)\psi}_1 & = & \lim_n \scal{U_2 \phi}{\pi_2 (A_n) U_2 \psi}_2 \\
& = & \lim_n \sum_{i,j} \scal{E_i \otimes I_\nn \otimes v_i}{F_j \otimes A_n \otimes w_j}_\SS \\
& = & \lim_n \sum_{i,j} \scal{v_i}{[\SS(E_i^\ast\odot_\mm F_j)] (A_n) w_j} \\
& = & \sum_{i,j} \scal{v_i}{[\SS(E_i^\ast\odot_\mm F_j)] (A) w_j} \\
& = & \scal{\phi}{\gg(A)\psi}_1
\end{eqnarray*}
as a consequence of weak*-continuity of $\SS(E_i^\ast\odot_\mm F_j)$. Hence, $\gg$ is normal, and, therefore, we have $\gg\in\cpn{\nn,\elle{\uhat_1}}$.

By Lemma 2.2 p.~139 in \cite{QTOS76} (or Corollary 10.4.14 in \cite{KadRin}), separability of $\uhat_1$ and normality of $\pi_1$ imply that there exists a (separable) Hilbert space $\uu$ such that $\uhat_1 = \hh\otimes\uu$ and $\pi_1 (B)  =  B \otimes I_\uu$ for all $B\in\lh$. We now prove that $\gg(A) \in \mm\votimes \elle{\uu}$ for all $A\in\nn$, i.e.~actually $\gg\in\cpn{\nn,\mm\votimes \elle{\uu}}$. By Proposition 1.6 p.~184 in \cite{Tak} and by von Neumann's  double commutation theorem (Theorem 3.9 p.~74 in \cite{Tak}), it is enough to show that $\gg(A) (B\otimes I_\uu) = (B\otimes I_\uu) \gg(A)$ for all $A\in\nn$ and $B\in\mm'$. So, for $\phi,\psi\in\lh\hotimes\kk$ with $\phi = \sum_{i=1}^n E_i \otimes v_i$, $\psi = \sum_{j=1}^m F_j \otimes w_j$, we have
\begin{eqnarray*}
\scal{\phi}{\gg(A) (B\otimes I_\uu) \psi}_1 & = & \scal{U_2\phi}{\pi_2 (A) U_2 \pi_1 (B) \psi}_1 \\
& = & \sum_{i,j} \scal{E_i \otimes I_\nn \otimes v_i}{BF_j \otimes A \otimes w_j}_\SS \\
& = & \sum_{i,j} \scal{v_i}{[\SS(E_i^\ast\odot_\mm BF_j)] (A) w_j} \\
& = & \sum_{i,j} \scal{v_i}{[\SS(E_i^\ast B \odot_\mm F_j)] (A) w_j} \\
& = & \scal{(B^\ast\otimes I_\uu)\phi}{\gg(A)\psi}_1 \, ,
\end{eqnarray*}
where the equality $E_i^\ast\odot_\mm BF_j = E_i^\ast B\odot_\mm F_j$ comes from item (3) of Proposition \ref{prop:prop. di odot}. Hence $\gg(A) \in (\mm'\votimes \C I_\uu)' = \mm\votimes\lu$, as claimed.

We conclude with the proof of Eq.~\eqref{eq. centr.}. If $E\in\lh$, $A\in\nn$ and $v,w\in\kk$, then we have, for $\ee = E^\ast \odot_\mm E$,
\begin{eqnarray*}
\scal{v}{\left[ \SS (\ee) \right] (A) w} & = & \scal{E\otimes I_\nn \otimes v}{E \otimes A \otimes w}_\SS \\
& = & \scal{U_2 \pi_1 (E) U_1 v}{\pi_2 (A) U_2 \pi_1 (E) U_1 w}_2 \\
& = & \scal{\pi_1 (E) U_1 v}{\gg(A) \pi_1 (E) U_1 w}_1 \\
& = & \scal{v}{U_1^\ast (E^\ast \otimes I_\uu) \gg(A) (E \otimes I_\uu) U_1 w} \\
& = & \scal{v}{U_1^\ast [(\ee \otimes \ii_\uu) \gg(A)] U_1 w} \, .
\end{eqnarray*}
Setting $U :=U_1$, we then obtain Eq.~\eqref{eq. centr.} in the special case $\ee = E^\ast \odot_\mm E$. 
The equality for generic $\ee\in\cp{\mm,\lh}$ then follows by Kraus Theorem \ref{Teo. Stines.} using normality of $\SS$ and of the  amplification supermap $\PI_\uu :  \ee \mapsto \ee \otimes \ii_\uu$. Finally,  linearity and Theorem \ref{CB = span CP} extend the equality to all $\ee \in \cb{\mm,\lh}$. This concludes the proof of Proposition \ref{teo. centr. prel.}.
\end{proof}

We still need another easy auxiliary lemma before coming to the proof of Theorem \ref{teo. centr.}.

\begin{theorem}{Lemma}\label{lemma:span}
Let $\kk$, $\vv$ be Hilbert spaces, and let $\ss$ be a linear subspace of $\kk\otimes\vv$. The following facts are equivalent:
\begin{itemize}
\item[{\rm (i)}] $\vv = \spannochiuso{(u^\ast\otimes I_\vv) v \mid v\in\ss \, , \, u\in\kk}$;
\item[{\rm (ii)}] the equality $\hh_0\otimes\vv = \spannochiuso{(A\otimes I_\vv)v \mid v\in\ss \, , \, A\in\elle{\kk,\hh_0}}$ holds for some Hilbert space $\hh_0$;
\item[{\rm (iii)}] the equality $\hh\otimes \vv = \spannochiuso{(A\otimes I_\vv) v \mid v\in\ss \, , \, A\in\elle{\kk,\hh}}$ holds for all Hilbert spaces $\hh$.
\end{itemize}
\end{theorem}
\begin{proof}
Clearly, condition {\rm (i)} implies {\rm (ii)} by taking $\hh_0 \equiv \C$, and condition {\rm (iii)} implies {\rm (i)} by taking $\hh \equiv \C$.

We now suppose that condition {\rm (ii)} holds, and prove {\rm (iii)}. If $\hh$ is a Hilbert space, let $\hhat = \spannochiuso{(A\otimes I_\vv) v \mid v\in\ss \, , \, A\in\elle{\kk,\hh}}$. Denote by $\hat{P}$ the orthogonal projection of $\hh\otimes \vv$ onto $\hhat$. Since $(\lh\otimes I_\vv)\hhat \subset \hhat$, the operator $\hat{P}$ commutes with $\lh\otimes I_\vv$, hence $\hat{P} = I_\hh\otimes P$ for some orthogonal projection $P$ of $\vv$ by Proposition 1.6 p.~184 in \cite{Tak}. Choose a Hilbert basis $\{ e_i \}_{i\in I}$ of $\hh_0$, and fix a vector $e\in\hh$ with $\no{e}=1$. Defining $E_i = ee_i^\ast \in \elle{\hh_0,\hh}$, we have $\sum_{i\in I} E_i^\ast E_i = I_{\hh_0}$, where the sum converges in the strong sense if $\# I = \infty$. It follows that, for all $A\in\elle{\kk,\hh_0}$ and $v\in\ss$,
\begin{eqnarray*}
(I_{\hh_0} \otimes P)(A\otimes I_\vv) v & = & \sum_{i\in I} (E_i^\ast \otimes I_\vv) (I_\hh \otimes P)(E_i A\otimes I_\vv) v \\
& = & \sum_{i\in I} (E_i^\ast \otimes I_\vv) \hat{P} (E_i A\otimes I_\vv) v \\
& = & \sum_{i\in I} (E_i^\ast \otimes I_\vv) (E_i A\otimes I_\vv) v \\
& = & (A\otimes I_\vv) v \, ,
\end{eqnarray*}
where convergence of the sum is in the norm topology of $\hh_0\otimes\vv$. 
By density, we conclude $I_{\hh_0} \otimes P = I_{\hh_0\otimes\vv}$, hence $P=I_\vv$, i.e.~$\hhat = \hh\otimes\vv$.
\end{proof}

We are now in position to prove Theorem \ref{teo. centr.}.

\begin{proof}(Proof of  Theorem \ref{teo. centr.})
The `if' part of the statement follows since $\SS = \CC_{\aa,\ff} \PI_\vv$, where $\aa\in\cpn{\elle{\kk_1}\votimes\lv,\elle{\kk_2}}$ is the quantum channel $\aa := V^\ast \odot_{\elle{\kk_1}\votimes\lv} V$, and $\CC_{\aa,\ff} : \cb{\mm_1\votimes\lv,\elle{\kk_1}\votimes\lv} \frecc \cb{\mm_2,\elle{\kk_2}}$ and $\PI_\vv : \cb{\mm_1,\elle{\kk_1}} \frecc \cb{\mm_1\votimes\lv,\elle{\kk_1}\votimes\lv}$ are the concatenation and amplification supermaps defined in Propositions \ref{compo} and \ref{ampli}, respectively.

Conversely, suppose $\SS\in\cpqn{\mm_1, \elle{\kk_1} ; \mm_2, \elle{\kk_2}}$. We can assume without restriction that $\mm_1\subset\elle{\hh_1}$ for some Hilbert space $\hh_1$. Let $\ell^2$ denote the Hilbert space of square-summable sequences, and define an isometry $T$ as follows
$$
T :  \kk_1 \to \kk_1 \otimes \ell^2 \, , \qquad   Tv = v \otimes e \, ,
$$
where $e \in \ell^2$ is a fixed vector with $\no{e} = 1$. Then, define two deterministic supermaps
\begin{align*}
\TT & : \cb{\mm_1 \votimes \elle{\ell^2} , \elle{\kk_1\otimes\ell^2}} \frecc \cb{\mm_1, \elle{\kk_1}} \\
\tilde{\SS} & : \cb{\mm_1 \votimes \elle{\ell^2}, \elle{\kk_1\otimes\ell^2}} \frecc \cb{\mm_2, \elle{\kk_2}}
\end{align*}
given by
$$
[\TT (\tilde{\ee})] (A) = T^\ast \tilde{\ee} (A \otimes I_{\ell^2}) T \quad \forall \tilde{\ee}\in \cb{\mm_1 \votimes \elle{\ell^2} , \elle{\kk_1\otimes\ell^2}} \, , \, \forall A\in\mm_1
$$
and
$$
\tilde{\SS} = \SS \TT \, .
$$
Since $\mm_1 \votimes \elle{\ell^2} \subset \elle{\hh_1\otimes \ell^2}$ and the Hilbert spaces $\hh_1\otimes \ell^2$ and $\kk_1\otimes \ell^2$ are isomorphic and infinite dimensional, we can apply Proposition \ref{teo. centr. prel.} to the deterministic supermap $\tilde{\SS}$ and obtain the existence of a Hilbert space $\uu$, an isometry $U:\kk_2\frecc \kk_1\otimes\ell^2\otimes\uu$ and a channel $\gg\in\cpn{\mm_2, \mm_1\votimes\elle{\ell^2}\votimes\lu}$ such that
$$
[\tilde{\SS}(\tilde{\ee})] (A)  =  U^\ast [(\tilde{\ee} \otimes \ii_\uu) \gg(A)] U \quad \forall \tilde{\ee} \in \cb{\mm_1 \votimes \elle{\ell^2}, \elle{\kk_1\otimes\ell^2}} , \forall A\in \mm_2 .
$$
On the other hand, we have $\TT (\ee\otimes\ii_{\ell^2}) = \ee$ for all $\ee\in\cb{\mm_1,\elle{\kk_1}}$ by directly inspecting the definition, hence $\tilde{\SS} (\ee\otimes\ii_{\ell^2}) = \SS(\ee)$. It follows that, for all $\ee\in\cb{\mm_1,\elle{\kk_1}}$ and $A\in\mm_2$,
\begin{align*}
[\SS (\ee)] (A) & = [\tilde{\SS} (\ee\otimes\ii_{\ell^2})] (A) \\
& = U^\ast [(\ee\otimes\ii_{\ell^2} \otimes \ii_\uu) \gg(A)] U \\
& = U^\ast [(\ee\otimes\ii_\ww) \gg(A)] U \, ,
\end{align*}
where we set $\ww : = \ell^2\otimes \uu$.

Now, let $\hat{\hh}_1$ be the following closed subspace of $\hh_1\otimes\ww$
\begin{equation}\label{eq:dens2}
\hat{\hh}_1 = \spannochiuso{(E\otimes I_\ww)Uv \mid v\in\kk_2 \, , \, E\in\elle{\kk_1,\hh_1}} \, ,
\end{equation}
and let $\hat{P}$ be the orthogonal projection of $\hh_1\otimes\ww$ onto $\hat{\hh}_1$. Since $(\elle{\hh_1}\otimes I_\ww) \hat{\hh}_1 \subset \hat{\hh}_1$, there is an orthogonal projection $P$ of $\ww$ such that $\hat{P} = I_{\hh_1}\otimes P$. Let $\vv=P\ww$, and define the operator $V:\kk_2\frecc \kk_1 \otimes \vv$ as $V := (I_{\kk_1} \otimes P) U$. From the fact that $\hat{P} = I_{\hh_1}\otimes P$, it clearly follows $\hat{P} (\hh_1\otimes \ww) = \hh_1\otimes \vv$ and $\hat{P} (E\otimes I_\ww) Uv = (E\otimes I_\vv) Vv$, so Eq.~\eqref{eq:dens2} can be turned into
$$
\hh_1\otimes\vv = \spannochiuso{(E\otimes I_\vv)Vv \mid v\in\kk_2 \, , \, E\in\elle{\kk_1,\hh_1}} \, .
$$
By Lemma \ref{lemma:span} (with $\ss \equiv V\kk_2$), we then have
$$
\vv = \spannochiuso{(u^\ast\otimes I_\vv)Vv \mid u\in \kk_1 \, , \, v\in\kk_2} \, .
$$
Define the quantum channel $\ff\in\cpn{\mm_2 , \mm_1\votimes \lv}$ given by
$$
\ff(A) := (I_{\hh_1} \otimes P) \gg(A) (I_{\hh_1} \otimes P^\ast) =(\ii_{\mm_1} \otimes \mathcal{P}) \gg (A) \quad \forall A\in\mm_2 \, ,
$$
with $\mathcal{P} := P\odot_{\elle{\ww}} P^\ast\in\cpn{\elle{\ww},\elle{\vv}}$.
Then, for $E\in\elle{\kk_1,\hh_1}$ and $\ee = E^\ast \odot_{\mm_1} E$,
\begin{eqnarray*}
[\SS(\ee)] (A) & = & U^\ast (E^\ast \otimes I_\ww) \gg(A) (E \otimes I_\ww) U \\
& = & U^\ast (E^\ast \otimes I_\ww) \hat{P}^\ast \hat{P} \gg(A) \hat{P}^\ast \hat{P} (E \otimes I_\ww) U \\
& = & V^\ast (E^\ast \otimes I_\vv) \ff(A) (E \otimes I_\vv) V \\
& = & V^\ast [(\ee \otimes \ii_\vv) \ff(A)] V
\end{eqnarray*}
for all $A\in\mm_2$. This equation extends to all $\ee\in\cb{\mm_1,\elle{\kk_1}}$ by the usual continuity and linearity argument. Finally, in order to show that $V$ is an isometry, pick a quantum channel $\ee\in\cpn{\mm_1,\elle{\kk_1}}$, and, since $\SS$ is deterministic,
$$
V^\ast V = V^\ast [(\ee \otimes \ii_\vv) \ff(I_{\mm_2})] V = [\SS (\ee)] (I_{\mm_2}) = I_{\kk_2} \, .
$$
This concludes the proof of Theorem \ref{teo. centr.}.
\end{proof}

We end the section with the proof of Proposition \ref{prop: minimality}. 

\begin{proof}(Proof of Proposition \ref{prop: minimality})
Define the linear space
$$
\vv_0 = \spanno{(u^\ast\otimes I_\vv)Vv \mid u\in\kk_1 \, , \, v\in\kk_2} \, ,
$$
and let $W:\vv_0 \frecc \vv'$ be the linear operator given by
$$
W(u^\ast\otimes I_\vv)Vv := (u^\ast\otimes I_{\vv'})V'v \, .
$$
We claim that $W$ is well defined and extends to an isometry $W:\vv \frecc \vv'$. As usual, we can assume with no restriction $\mm_1\subset\elle{\hh_1}$. Pick then a vector $e\in\hh_1$ with $\no{e}=1$. For all $u,w\in\kk_1$, define
$$
\ee_{u,w} := (ue^\ast) \odot_{\mm_1} (ew^\ast) \in \cb{\mm_1,\elle{\kk_1}} \, .
$$
If $\phi\in\vv_0$, with $\phi = \sum_{i=1}^n (u_i^\ast\otimes I_\vv)Vv_i$, we have
\begin{eqnarray*}
\no{W\phi}^2 & = & \sum_{i,j} \scal{v_j}{V^{\prime\ast}(u_j u_i^\ast \otimes I_{\vv'}) V' v_i} \\
& = & \sum_{i,j} \scal{v_j}{[\SS(\ee_{u_j,u_i})] (I_{\mm_2}) v_i} \\
& = & \sum_{i,j} \scal{v_j}{V^\ast (u_j u_i^\ast \otimes I_{\vv}) V v_i} \\
& = & \no{\phi}^2 \, .
\end{eqnarray*}
Thus, $W$ is well defined and isometric, and extends to an isometry $W:\vv \frecc \vv'$ by density of $\vv_0$ in $\vv$.

For all $u\in\kk_1$, $v\in\kk_2$ and $w\in\vv'$, we have
\begin{align*}
\scal{u\otimes w}{(I_{\kk_1} \otimes W)Vv} & = \scal{w}{(u^\ast\otimes I_{\vv'})(I_{\kk_1} \otimes W)Vv} \\
& = \scal{w}{W(u^\ast \otimes I_\vv)Vv} \\
& = \scal{w}{(u^\ast \otimes I_{\vv'})V'v} \\
& = \scal{u\otimes w}{V'v} \, ,
\end{align*}
hence $(I_{\kk_1} \otimes W)V = V'$.

If $E,F\in\elle{\kk_1,\hh_1}$ and $v,w\in\kk_2$, then, for all $A\in\mm_2$,
\begin{align*}
& \scal{(E\otimes I_\vv)Vv}{\ff(A) (F\otimes I_\vv)Vw} = \scal{v}{[\SS(E^\ast \odot_{\mm_1} F)] (A) w} \\
& \qquad \qquad \qquad = \scal{(E\otimes I_{\vv'})V'v}{\ff'(A) (F\otimes I_{\vv'})V'w} \\
& \qquad \qquad \qquad = \scal{(E\otimes I_\vv)Vv}{(I_{\hh_1} \otimes W^\ast) \ff'(A) (I_{\hh_1} \otimes W) (F\otimes I_\vv)Vw} \, .
\end{align*}
By the minimality condition \eqref{dens. in hhat1} and Lemma \ref{lemma:span}, we have
$$
\hh_1\otimes \vv = \spannochiuso{(E\otimes I_\vv)Vv \mid v\in\kk_2 \, , \, E\in\elle{\kk_1,\hh_1}} \, ,
$$
hence the last equation shows that $\ff(A) = (I_{\mm_1} \otimes W^\ast) \ff'(A) (I_{\mm_1} \otimes W)$.

We finally come to uniqueness of $W$. Suppose that $U : \vv\frecc \vv'$ is another isometry such that $(I_{\kk_1} \otimes U)V = V'$. Then, for all $u\in\kk_1$, $v\in\kk_2$ and $w\in\vv'$,
\begin{eqnarray*}
\scal{w}{U(u^\ast\otimes I_\vv) V v} & = & \scal{u\otimes w}{(I_{\kk_1} \otimes U) V v} \\
& = & \scal{u\otimes w}{V' v} \\
& = & \scal{w}{(u^\ast\otimes I_\vv) V' v} \\
& = & \scal{w}{W (u^\ast\otimes I_\vv) V v} \, ,
\end{eqnarray*}
i.e.~$U(u^\ast\otimes I_\vv) Vv = W(u^\ast\otimes I_\vv) Vv$. By the minimality condition \eqref{dens. in hhat1}, $U=W$.
\end{proof}

\subsection{An application of Theorem \ref{teo. centr.}:  transforming a quantum measurement into a quantum channel}\label{subsect:meastochan}
For simplicity we consider here quantum measurements with a countable set of outcomes, denoted by $X$. In the algebraic language,  a  measurement   on the quantum system with  Hilbert space $\kk_1$ and   with outcomes in  $X$  is described by a quantum channel $\ee \in \cpn{\mm_1,  \elle {\kk_1}}$,  where $\mm_1 \equiv \ell^\infty (X)$  is the von Neumann algebra of  the bounded complex functions (i.e.~sequences) on $X$ with uniform norm 
$\no{f}_{\infty}  : =  \sup_{i \in  X}  |f_i|$.
The channel $\ee$ maps the function $f \in \ell^\infty (X)$  into the operator 
\begin{equation}\label{eq:CP1=POVM}
\ee (f)  =   \sum_{i \in  X}    f_i   \, P_i  \in  \elle{\kk_1} \, ,  
\end{equation}   
where each $P_i$  is a non-negative operator in $\elle  {\kk_1}$ and $\sum_{i \in  X}  P_i  =  I_{\kk_1}$. Note that the map $i\mapsto P_i$ is a normalized {\em positive operator valued measure (POVM)} based on the discrete space $X$ and with values in $\elle{\kk_1}$. Actually, Eq.~\eqref{eq:CP1=POVM} allows us to identify the convex set of measurements $\cpn{\ell^\infty(X),  \elle {\kk_1}}$ with the set of {\em all} normalized $\elle{\kk_1}$-valued POVMs on $X$.\footnote{Indeed, by commutativity of $\ell^\infty (X)$ the set $\cpn{\ell^\infty (X),  \elle{\kk_1}}$ coincides with the set of all normalized weak*-continuous {\em positive} maps from $\ell^\infty (X)$ into $\elle{\kk_1}$ (Theorem 3.11 in \cite{Paul}). The latter set is just the set of all normalized $\elle{\kk_1}$-valued POVMs on $X$, the identification being the one given in Eq.~\eqref{eq:CP1=POVM}.}

The probability of  obtaining the outcome $i \in  X$ when the measurement is performed on a system prepared in the quantum state $\rho\in \trcl{\kk_1}$  ($\rho  \ge0$,  $\trt{\rho}=1$) is given by  the Born rule
\begin{align*}
p_i  =     \trt{\rho   P_i} \, ,
\end{align*}
and the expectation value of the function $f \in \ell^\infty (X)$ with respect to the probability distribution $p$ is given by  
\begin{align*}
\mathbb E_{p}  (f)  :=   \sum_{i\in X}   p_i  f_i  =    \trq{\rho  \ee (f)} \, .  
\end{align*}  
The above equation allows us to interpret the channel $\ee$ as an \emph{operator valued expectation}  (see e.g.~\cite{dirk}).  

Now, consider the deterministic supermaps sending quantum measurements in the set $\cp  {\ell^\infty (X),  \elle{\kk_1} }$ to quantum operations in $\cp {\mm_2,  \elle {\kk_2}}$, where  $\mm_2   \equiv \elle {\hh_2}  $. Our dilation Theorem  \ref{teo. centr.} (in the predual form of Remark \ref{rem:teo.centr.pred.}) states that every deterministic supermap  $\SS : \cb{\ell^\infty (X) , \elle{\kk_1}} \frecc \cb{\elle{\hh_2} , \elle{\kk_2}}$ is of the form
\begin{equation}\label{eq:non so}
[\SS (\ee)]_\ast (\rho) = \ff_\ast [(\ee\otimes \ii_{\vv})_\ast (V\rho V^\ast)] \quad \forall \ee\in\cb{\ell^\infty (X) , \elle{\kk_1}} \, , \, \forall \rho\in\trcl{\kk_2} \, ,
\end{equation}
where  $\vv$ is a Hilbert space, $V:\kk_2\frecc \kk_1 \otimes \vv$ an isometry and  $\ff\in\cpn{\elle{\hh_2},\ell^\infty (X)\votimes\lv}$ a quantum channel. In our case, we have the identification
$$
\ell^\infty (X) \votimes \lv \simeq \ell^\infty (X ; \lv) \, ,
$$
where $\ell^\infty (X ; \lv)$ is the von Neumann algebra of the bounded $\lv$-valued functions on $X$. Its predual space is
$$
(\ell^\infty (X) \votimes \lv)_\ast \simeq \ell^1 (X ; \trcl{\vv}) \, ,
$$
i.e.~the space of norm-summable sequences with index in $X$ and values in the Banach space of the trace class operators on $\vv$ (see Theorem 1.22.13 in \cite{Sakai}).
%, with norm $\no{(A_i)}_{\ell^1} = \sum_{i\in\Omega} \no{A_i}_\infty < \infty$. 
%Since in our case we have the identification  
%\begin{align*} 
%\mm_1 \votimes \lv  \simeq   \bigoplus_{i\in  \mathsf X }     \elle {\vv_i} \, ,  \qquad \vv_i  \simeq  \vv \quad \forall  i \in \mathsf X,
%\end{align*}  
%the channel $\ff\in\cpn{\mm_2,\mm_1\votimes\lv}$ can be written in the form 
%\begin{align*}
%\ff   (A)   =      \bigoplus_{i\in \mathsf X} \ff_{i}  (A)  \quad \forall A  \in  \elle{\hh_2} , 
%\end{align*}
%where  $\ff_{i}  \in  \cp {\elle{ \hh_2 }, \elle{\vv}}$  is a channel, i.e. $  \ff_{i}  (I_{\kk_2})  =  I_{\vv}$ for every $i \in \mathsf X$.  
In the Schr\"odinger picture, the channel $\ff_\ast$  can be realized by first reading the classical information  carried by the system with algebra  $\ell^\infty (X)$ and, conditionally to the value $i \in  X$, by performing the quantum channel $\ff_{i\,\ast} : \trcl{\vv} \frecc \trcl{\hh_2}$ given by
$$
\ff_{i\,\ast} (\sigma) = \ff_\ast (\delta_i\, \sigma) \quad \forall\sigma \in \trcl{\vv} \, ,
$$
where $\delta_i\, \sigma\in\ell^1 (X ; \trcl{\vv})$ is the sequence $(\delta_i\, \sigma)_k = \delta_{ik} \, \sigma \ \forall k\in X$, $\delta_{ik}$ being the Kronecker delta. Hence, Eq.~\eqref{eq:non so} can be rewritten as
$$
[\SS (\ee)]_\ast (\rho) = \sum_{i\in X} \ff_{i\,\ast} [(\ee\otimes \ii_{\vv})_\ast (V\rho V^\ast)_i] \, .
$$
In other words, Theorem \ref{teo. centr.} states that the most general transformation of a quantum measurement on $\kk_1$ into a quantum channel from states on $\kk_2$  to states on  $\hh_2$ can be realized by 
\begin{enumerate} 
\item applying an invertible dynamics (the isometry  $V)$ that transforms the input system $\kk_2$  into the composite system  $\kk_1  \otimes \vv$, where $\vv$ is an ancillary system;
\item performing the given measurement ($\ee_*$, in the predual picture) on $\kk_1$, thus obtaining the outcome $i\in  X$;
\item conditionally to the outcome $i \in X$, applying a physical transformation  (the channel  $\ff_{i\, \ast}$) on  the ancillary system $\vv$, thus converting it into the output system $\hh_2$. 
\end{enumerate}

\section{Radon-Nikodym derivatives of supermaps}\label{sez. Radon}

The dilation theorem for deterministic supermaps will be generalized here to probabilistic supermaps. 
In this case,  the following theorem provides an analog of the Radon-Nikodym theorem for CP maps (compare with \cite{Arveson,BelStasz}, and see also \cite{Raginsky} for the particular case of quantum operations).

\begin{theorem}{Theorem}\label{teo. Radon-Nicodym}
{\rm (Radon-Nikodym theorem for supermaps)}
Suppose that $\SS$ is a deterministic supermap in $\cpqn{\mm_1,\elle{\kk_1} ; \mm_2,\elle{\kk_2}}$ and let $(\vv , \, V, \, \ff)$ be its minimal dilation. If $\TT \in\cpq{\mm_1,\elle{\kk_1} ; \mm_2,\elle{\kk_2}}$ is such that $\TT \ll \SS$, then there exists a unique element $\gg\in\cp{\mm_2, \mm_1\votimes \lv}$ with $\gg \preceq \ff$ and such that
\begin{equation}\label{radnic}
\left[\TT (\ee)\right]  (A)= V^\ast [(\ee \otimes \ii_\vv) \gg(A)] V \quad \forall \ee \in \cb{\mm_1,\elle{\kk_1}} \, , \, \forall A\in\mm_2 \, .
\end{equation}
\end{theorem} 
\begin{proof}
Without loss of generality,  let us suppose $\mm_1 \subset \elle{\hh_1}$ for some suitable Hilbert space $\hh_1$. Hence, we can regard the quantum channel $\ff$ as an element in $\cpn{\mm_2, \elle{\hh_1 \otimes \vv}}$.  Consider the  Stinespring dilation of the channel $\ff $, given by
$$
\ff(A) = U^\ast \pi(A) U  %= [(U^\ast\odot_{\elle{\uu}} U) \,\pi] (A)
 \quad \forall A\in\mm_2,
$$
where $U: \hh_1\otimes\vv \frecc \uu$ is an isometry,  $\uu$ is a separable Hilbert space, and $\pi : \mm_2 \frecc \elle{\uu}$ is a normal unital $\ast$-homomorphism (see e.g.~Theorem 2.1 p.~137 of \cite{QTOS76}). 
In particular, we can take the minimal Stinespring dilation, which satisfies the relation
$$
\uu = \spannochiuso{\pi(A) U u \mid A\in\mm_2 \, , \, u\in\hh_1\otimes\vv} \, .
$$

Let us define the dense subset $\uu_0  \subset \uu$ as  
$$
\uu_0 :=\spanno{\pi(A) U   u \mid  A   \in  \mm_2 \, , \, u  \in \hh_0},
$$
where $\hh_0$ is the following dense subset of $\hh_1\otimes\vv$
\begin{equation}\label{eq:ins.denso}
\hh_0  :=  \spanno{(E\otimes I_\vv) Vv \mid E\in\elle{\kk_1, \hh_1} \, , \, v\in\kk_2}
\end{equation}
(see Eq.~\eqref{dens. in hhat1} and Lemma \ref{lemma:span} for the proof that $\hh_0$ is dense in $\hh_1\otimes\vv$).   

We now introduce a positive sesquilinear form  $\scal{\cdot}{\cdot}_0$ on $\uu_0$, which we will briefly see that is bounded and thus can be extended by continuity to a form on $\uu$. If $\phi = \sum_{i=1}^n \pi(A_i) U (E_i\otimes I_\vv) Vv_i$ and $\psi = \sum_{j=1}^m \pi(B_j) U (F_j\otimes I_\vv) Vw_j$ are two generic elements in $\uu_0$, define
\begin{align*}
\scal{\phi}{\psi}_0 & := \sum_{i,j} \scal{v_i}{[\TT ( E_i^\ast \odot_{\mm_1} F_j)] (A_i^\ast B_j) w_j} \\
& = \scal{\sum_i E_i\otimes A_i \otimes v_i}{\sum_j F_j\otimes B_j \otimes w_j}_\TT \, .
\end{align*}
We claim that $\scal{\cdot}{\cdot}_0$ is a well defined positive and bounded sesquilinear form on $\uu_0$. In order to show this, it is enough to prove that
$$
0\leq \scal{\phi}{\phi}_0 \leq \no{\phi}^2 \quad \forall \phi\in\uu_0 \, .
$$
Indeed, the first inequality is clear from Lemma \ref{prop. sulla forma ass.}. For the second, again by Lemma \ref{prop. sulla forma ass.} we have, for $\phi = \sum_{i=1}^n \pi(A_i) U (E_i\otimes I_\vv) Vv_i$,
\begin{align*}
\scal{\phi}{\phi}_0 & = \scal{\sum_i E_i\otimes A_i \otimes v_i}{\sum_j E_j\otimes A_j \otimes v_j}_\TT \\
& \leq \scal{\sum_i E_i\otimes A_i \otimes v_i}{\sum_j E_j\otimes A_j \otimes v_j}_\SS \\
& = \sum_{i,j} \scal{v_i}{[\SS(E_i^\ast \odot_{\mm_1} E_j)](A_i^\ast A_j) v_j} \\
& =  \sum_{i,j} \scal{v_i}{V^\ast (E_i^\ast \otimes I_\vv) \ff (A_i^\ast A_j)  (E_j \otimes I_\vv) V v_j} \\
& =  \sum_{i,j} \scal{v_i}{V^\ast (E_i^\ast \otimes I_\vv) U^\ast \pi(A_i^\ast A_j) U (E_j \otimes I_\vv) V v_j} \\
& = \scal{\sum_i \pi(A_i)U(E_i \otimes I_\vv) Vv_i}{\sum_j \pi(A_j) U (E_j \otimes I_\vv) V v_j} \\
& = \no{\phi}^2 \, .
\end{align*}
This concludes the proof of our claim.

We continue to denote by $\scal{\cdot}{\cdot}_0$ the previous form extended by continuity to the whole space $\uu$. Then, there exists a bounded operator $C\in\elle{\uu}$, with $0\leq C \leq I_\uu$, such that
$$
\scal{\phi}{\psi}_0 = \scal{\phi}{C\psi} \quad \forall \phi ,\psi\in\uu \, .
$$
Note that $C$ commutes with the von Neumann algebra $\mm_\pi  : = \pi  (\mm_2)$.\footnote{The linear space $\pi (\mm_2)$ is a von Neumann algebra in $\elle{\uu}$ by Proposition 3.12 p.~136 in \cite{Tak}.} Indeed, if $\phi=  \sum_{i=1}^n \pi(A_i) U (E_i\otimes I_\vv) V v_i $ is a generic element in $\uu_0$, then, for all $A\in\mm_2$,
\begin{align*}
\scal{\phi}{C\pi(A)\phi} & = \scal{\phi}{\pi(A)\phi}_0 \\  
%&=  \sum_{i,j} \scal{\pi(A_i) U (E_i\otimes I_\vv) Vv_i}{\pi(D A_j) U (E_j\otimes I_\vv) Vv_j}_0 \\
& = \sum_{i,j} \scal{v_i}{[\TT(E_i^\ast \odot_{\mm_1} E_j)] (A_i^\ast A  A_j) v_j} \\
%& = \sum_{i,j} \scal{\pi(D^\ast A_i) U (E_i\otimes I_\vv) Vv_i}{\pi(A_j) U (E_j\otimes I_\vv) Vv_j}_0 \\
& = \scal{\pi(A^\ast) \phi}{\phi}_0 \\
& = \scal{\pi(A)^\ast \phi}{C\phi}\\
& = \scal{ \phi}{\pi(A) C\phi} .
\end{align*}
By density and the polarization identity we then obtain $C\pi(A) = \pi(A) C $ for all $A \in  \mm_2$.

We are now ready to define the map $\gg\in\cp{\mm_2,\elle{\hh_1\otimes\vv}}$ as
$$
\gg := (U^\ast \odot_{\elle{\uu}} U) (C^{\frac12} \odot_{\mm_\pi} C^{\frac12}) \, \pi \, .
$$

For all $E,F\in\elle{\kk_1,\hh_1}$, $A\in\mm_2$ and $v,w\in\kk_2$, we have   
\begin{align*}
\scal{v}{[\TT(E^\ast\odot_{\mm_1} F)](A) w} &  =  \scal{U(E\otimes I_\vv)Vv}{\pi(A)U(F\otimes I_\vv)Vw}_0 \\
&  =  \scal{U(E\otimes I_\vv)Vv}{C\pi(A)U(F\otimes I_\vv)Vw} \\
&  =  \scal{v}{V^\ast (E^\ast\otimes I_\vv) U^\ast C^{\frac 12} \pi(A) C^{\frac 12} U(F\otimes I_\vv)Vw} \\
&  =  \scal{v}{V^\ast (E^\ast\otimes I_\vv) \gg (A) (F\otimes I_\vv)Vw} \, . 
\end{align*}
Since $v,w \in\kk_2$ are arbitrary, we just proved the relation 
\begin{equation}\label{punto chiave}
[\TT(E^\ast\odot_{\mm_1} F)](A) = V^\ast (E^\ast\otimes I_\vv) \gg(A)(F\otimes I_\vv) V 
\end{equation}
for all $E,F\in \elle{\kk_1,\hh_1}$ and $A \in  \mm_2$.

Eq.~\eqref{punto chiave} allows us to prove that the range of the map $\gg$ is contained in  $\mm_1\votimes\lv$, i.e., that for all $A\in\mm_2$ we have $\gg (A) \in \mm_1\votimes\lv$.  To prove this, it is enough to show that $\gg(A)$ commutes with $\left( \mm_1\votimes\lv\right)^\prime  =  \mm_1^\prime \votimes \C I_\vv$. Indeed, for every $B\in\mm_1'$ and for a generic element $u= \sum_{i=1}^n (E_i\otimes I_\vv) V v_i  \in \hh_0$  we have 
\begin{align*}
\scal{u} { \gg(A) (B\otimes I_\vv)  u }  &=    
 \sum_{i,j} \scal{(E_i\otimes I_\vv) V v_i}{\gg(A)(B\otimes I_\vv)(E_j\otimes I_\vv) V v_j}\\
 &=  \sum_{i,j} \scal{  v_i}{ V^* (E_i^*\otimes I_\vv)\gg(A)(B E_j\otimes I_\vv) V v_j}\\
& = \sum_{i,j} \scal{v_i}{[\TT(E_i^\ast\odot_{\mm_1} BE_j)](A) v_j} \\
&=\sum_{i,j} \scal{v_i}{[\TT(E_i^\ast B\odot_{\mm_1} E_j)](A) v_j} \\
&=  \sum_{i,j} \scal{  v_i}{ V^* (E_i^*  B\otimes I_\vv)\gg(A)(E_j\otimes I_\vv) V v_j}\\
&= \sum_{i,j}\scal{(E_i\otimes I_\vv) V v_i}{(B\otimes I_\vv)\gg(A)(E_j\otimes I_\vv) V v_j}\\
&= \scal{u}{(B\otimes I_\vv) \gg(A)u} \, ,
\end{align*}
where the equality $E_i^\ast\odot_{\mm_1} BE_j = E_i^\ast B \odot_{\mm_1} E_j$ is item (3) of Proposition \ref{prop:prop. di odot}.
By the polarization identity and by density of $\hh_0$ in $\hh_1\otimes\vv$ we then obtain  $\gg(A) (B\otimes I_\vv) = (B\otimes I_\vv) \gg(A)$.

Moreover, from Eq.~\eqref{punto chiave} the desired relation Eq.~\eqref{radnic} easily follows:
indeed, Eq.~\eqref{punto chiave} proves Eq.~\eqref{radnic} for all $\ee\in\elle{\hh_1,\kk_1} \odot_{\mm_1} \elle{\kk_1,\hh_1}$; Eq.~\eqref{radnic} for all $\ee\in\cb{\mm_1,\elle{\kk_1}}$ then follows by linearity and normality of $\TT$ and Theorems \ref{CB = span CP} and \ref{Teo. Stines.}.

Note that Eq.~\eqref{radnic} determines $\gg$ uniquely:  if $\gg'  \in \cp{ \mm_2,    \elle{ \mm_1 \votimes \vv} }$ is a map satisfying  Eq.~\eqref{radnic}, then   for a generic element $u  \in  \hh_0$, written as $u= \sum_{i=1}^n  (E_i \otimes I_{\vv})  V v_i $,  we must have
\begin{align*}
\scal{u} {  \gg' (A)  u}  &  =    \sum_{i,j}  \scal{v_i}{   V^\ast (E_i^\ast\otimes I_\vv) \gg'(A)(E_j\otimes I_\vv) V   v_j} \\ 
&  =    \sum_{i,j}  \scal{v_i}{   V^*  \left[\TT (E_i^* \odot E_j)\right]  (A)   V v_j} \\
 &=    \sum_{i,j}  \scal{v_i}{   V^\ast (E_i^\ast\otimes I_\vv) \gg(A)(E_j\otimes I_\vv) V   v_j} \\ 
&  = \scal{u} {  \gg (A)  u}, 
\end{align*}
which, by polarization identity and by density of $\hh_0$, implies $\gg'(A) =  \gg(A)$ for every $A \in \mm_2$, and therefore $\gg' = \gg$.

To conclude we prove that the map $\gg$ has the property $\gg\preceq \ff$:  
\begin{align*}
\gg &= (U^\ast \odot_{\elle{\uu}} U) (C^{\frac12} \odot_{\mm_\pi} C^{\frac12}) \, \pi \\
 & \preceq   (U^\ast \odot_{\elle{\uu}} U)   \, \pi \\
 &  =  \ff,
\end{align*} 
the inequality $C^{\frac12} \odot_{\mm_\pi} C^{\frac12} \preceq \ii_{\mm_\pi,\,\elle{\uu}}$ being item (4) of Proposition \ref{prop:prop. di odot}.   
\end{proof}

\begin{definition}{Definition}
In Theorem \ref{teo. Radon-Nicodym}, the CP map $\gg\in\cp{\mm_2, \mm_1\votimes \lv}$ defined by Eq.~\eqref{radnic} is the {\em Radon-Nikodym derivative} of the supermap $\TT$ with respect to $\SS$.
\end{definition}

\begin{definition}{Remark}
Note that the validity of Theorem \ref{teo. Radon-Nicodym}  can be trivially extended to quantum supermaps that are bounded by positive multiples of deterministic supermaps, i.e.~to supermaps $\TT$ such that $\TT \ll   \lambda \SS$ for some positive $\lambda \in \R$ and some deterministic supermap $\SS$.  
\end{definition}

\section{Superinstruments}\label{sez. superstr.}
Here we apply the Radon-Nikodym theorem proven in the previous section to the study of \emph{quantum superinstruments}.  Quantum superinstruments describe measurement processes where the measured object is not a quantum system, as in ordinary instruments, but rather a quantum device.  While ordinary quantum instruments are defined as probability measures with values in the set of quantum operations (see \cite{DavLew}, and also \cite{QTOS76} for a more complete exposition), quantum superinstruments are defined as probability measures with values in the set of quantum supermaps.

\begin{definition}{Definition}
Let $\Omega$ be a measurable space with $\sigma$-algebra $\sigma (\Omega)$ and let $\SS$ be a map from $\sigma (\Omega)$ to $\cpq{\mm_1,\elle{\kk_1}; \mm_2, \elle{\kk_2}}$, sending the measurable set $B \in \sigma (\Omega)$ to the supermap $\SS_B  \in \cpq{\mm_1,\elle{\kk_1}; \mm_2, \elle{\kk_2}}$. We say that $\SS$ is a {\em quantum superinstrument} if it satisfies the following properties:
\begin{enumerate}
\item[{\rm (i)}] $\SS_\Omega$ is deterministic;
\item[{\rm (ii)}] if $n\in\N\cup\{\infty\}$ and $B = \bigcup_{i=1}^n  B_i$ with $B_i \cap B_j = \emptyset$ for $i\neq j$, then $\SS_B = \sum_{i=1}^n  \SS_{B_i}$, where if $n = \infty$ convergence of the series is understood in the following sense:
$$
[\SS_B  (\ee)] (A) = \wklim_k \sum_{i=1}^k [\SS_{B_i}  (\ee)] (A) \quad \forall \ee\in\cb{\mm_1,\elle{\kk_1}} , \forall A\in\mm_2 .
$$
\end{enumerate}
\end{definition}
We will briefly see that every quantum superinstrument is associated to an ordinary quantum instrument in a unique way. Before giving the precise statement, we recall the notion of quantum instrument, which is central in the statistical description of quantum measurements:
\begin{definition}{Definition}
A map $\jj:\sigma (\Omega) \frecc \cp{\mm,\nn}$  is a \emph{quantum instrument} if it satisfies the following properties:
\begin{enumerate}
\item[{\rm (i)}] $\jj_\Omega$ is a quantum channel; 
\item[{\rm (ii)}] if $n\in\N\cup\{\infty\}$ and $B = \bigcup_{i=1}^n  B_i$ with $B_i \cap B_j = \emptyset$ for $i\neq j$, then $\jj_B = \sum_{i=1}^n  \jj_{B_i}$, where if $n = \infty$ convergence of the series is understood in the following sense:
$$
\jj_B (A) = \wklim_k \sum_{i=1}^k \jj_{B_i} (A) \quad \forall A\in\mm \, .
$$ 
\end{enumerate}
\end{definition}

We then have the following dilation theorem for quantum superinstruments.
\begin{theorem}{Theorem}\label{Osawa}
{\rm (Dilation of quantum superinstruments)}
Suppose that $\SS: \sigma (\Omega) \frecc \cpq{\mm_1,\elle{\kk_1}; \mm_2, \elle{\kk_2}}$ is a quantum superinstrument and let $(\vv, \, V, \, \ff)$ be the minimal dilation of the deterministic supermap $\SS_\Omega $.  Then there exists a unique quantum instrument  $\jj:  \sigma (\Omega) \frecc  \cp{\mm_2,\mm_1 \votimes \lv}$ such that
\begin{equation}\label{postselect}
[\SS_B (\ee)]  (A)  =  V^\ast [(\ee\otimes \ii_\vv)   \jj_B (A)] V \quad \forall \ee\in\cb{\mm_1,\elle{\kk_1}} \, , \, \forall A\in\mm_2
\end{equation}
for all $B\in\sigma(\Omega)$.
\end{theorem}
\begin{proof}  
Let $B \in \sigma(\Omega)$ be an arbitrary measurable set.   By additivity of the measure $\SS$, we have $\SS_{\Omega} = \SS_B  +  \SS_{\Omega\setminus B}$, that is, $\SS_B \ll \SS_\Omega$. Let $(\vv,V,\ff)$ be the minimal dilation of $\SS_\Omega$. By Theorem \ref{teo. Radon-Nicodym}, Eq.~\eqref{postselect} holds for some uniquely defined $\jj_B \in \cp{\mm_2,\mm_1\votimes\elle{\vv}}$, with $\jj_B \preceq \ff$. Clearly, for $B = \Omega$ one has $\jj_{\Omega} = \ff$, hence $\jj_{\Omega}$ is a quantum channel. Now, suppose that $n\in\N\cup\{\infty\}$ and $B = \bigcup_{i=1}^{n}  B_i$, with $B_i \cap B_j = \emptyset$ for $i\neq j$. If $n\in\N$, the equality $\jj_{\cup_{i=1}^{n}  B_i} = \sum_{i=1}^n  \jj_{B_i}$ easily follows by additivity of the superinstrument $\SS$ and uniqueness of the Radon-Nikodym derivative. If $n=\infty$, then the sequence of CP maps $\gg_n =\sum_{i=1}^n  \jj_{B_i}$ is CP-increasing and CP-bounded, since $\gg_n=\jj_{\cup_{i=1}^{n}  B_i}\preceq \ff$. Therefore, we have $\gg_n \Uparrow \gg_\infty$ for some $\gg_\infty \in \cp{\mm_2,\mm_1\votimes\elle{\vv}}$. We prove that $\gg_{\infty} = \jj_{B}$. Indeed, for every $\ee \in \cb{\mm_1,\elle{\kk_1}} $ and $A\in\mm_2$, by Proposition \ref{Teo. Berb. 2} we have
\begin{align*}
V^\ast [(\ee\otimes \ii_{\vv})   \gg_{\infty} (A)] V & = \wklim_n \sum_{i=1}^n V^\ast [(\ee\otimes \ii_{\vv}) \jj_{B_i} (A)] V \\
& = \wklim_n \sum_{i=1}^n [\SS_{B_i} (\ee)] (A) \\
& = [\SS_B (\ee)] (A) \, .
\end{align*}
By uniqueness of the Radon-Nikodym derivative we then conclude $\gg_{\infty} = \jj_B$.
\end{proof}

The physical interpretation of the dilation of quantum superinstruments is clear in the Schr\"odinger picture. Indeed, taking the predual of Eq.~\eqref{postselect}, we have for all $\rho \in \trcl{\kk_2}$ and $\ee \in \cb{\mm_1,\elle{\kk_1}}$
$$
[\SS_B (\ee)]_\ast (\rho) =  \jj_{B\, \ast} \left[(\ee \otimes \ii_\vv)_\ast   ( V \rho V^\ast)  \right] \, .
$$
This means that the system with Hilbert space $\kk_2$   (initially prepared in the quantum state $\rho$) undergoes an invertible evolution, given by the isometry $V$, that transforms it into the composite system with Hlbert space $ \kk_1  \otimes \vv $; then the system with Hilbert space $\kk_1$ is transformed by means of the quantum channel $\ee_*$, while nothing is done on the ancilla; finally, the quantum measurement described by the instrument $\jj_*$ is performed jointly on the system and ancilla.

\subsection{Application of Theorem \ref{Osawa}: Measuring a measurement}\label{subsect:measmeas}

Suppose that we want to characterize some property of a quantum measuring device on a system with Hilbert space $\kk_1$:  For example, we may have a device performing a projective measurement on an unknown orthonormal basis, and we may want to find out the basis.   In this case  the set of possible answers to our question is thus the set of all orthonormal bases. In a more abstract setting, the possible outcomes will constitute a measure space $\Omega$ with $\sigma$-algebra $\sigma (\Omega)$.  This includes also the case of full tomography of the measuring device \cite{lor,fiur,soto,eisert}, in which the outcomes in $\Omega$ label all possible  measuring devices.  The mathematical object describing our task will be a superinstrument  taking the given measurement as input and yielding an outcome in the set $B \in \sigma(\Omega)$ with some probability.   In the algebraic framework, we will describe the input measurement as a quantum channel $\ee \in \cp {  \mm_1, \elle{\kk_1}}$, where $\mm_1 \equiv \ell^\infty (X)$ is the algebra of the complex bounded functions on $X$ (see the discussion in Section \ref{subsect:meastochan}). 

\subsubsection{Outcome statistics for a measurement on a measuring device} 
If we only care about the outcomes in $\Omega$ and their statistical distribution, then the output of the superinstrument will be trivial, that is  $  \mm_2  \equiv  \elle {\kk_2}  \equiv \C$.  
In this case, Theorem \ref{Osawa} states that every superinstrument $\SS: \sigma (\Omega) \frecc \cpq{\ell^\infty(X),\elle{\kk_1}; \mathbb C, \mathbb C}$ will be of the form 
$$
\SS_B (\ee) =  \scal{v}{(\ee\otimes \ii_\vv) (\jj_B) v} \quad  \forall \ee\in\cb{\ell^\infty (X) , \elle{\kk_1}} \, , \, B \in \sigma(\Omega) \, ,
$$
where $\vv$ is an ancillary Hilbert space, $v\in\kk_1\otimes \vv$ is a unit vector, and $\jj:  \sigma (\Omega) \frecc \cp{\C,\ell^\infty (X) \votimes \lv} \simeq \ell^\infty (X;\lv)_+$ is just a weak*-countably additive positive measure on $\Omega$ with values in $\ell^\infty (X;\lv)$, satisfying $(\jj_\Omega)_i = I_\vv \ \forall i\in X$. Note that in this case each supermap $\SS_B$ is actually a linear map $\SS_B : \cb{\ell^\infty (X) , \elle{\kk_1}} \frecc \C$, and, if $\ee$ is a quantum channel, the map $B\mapsto \SS_B (\ee)$ is a probability measure on $\Omega$. In the Schr\"odinger picture
\begin{equation}\label{eq:non so2}
\SS_B (\ee) = [\jj_{B\,\ast} (\ee\otimes \ii_\vv)_\ast] (\omega_v) \, , 
\end{equation}
where $\omega_v$ is the state in $\trcl{\kk_1\otimes\vv}$ given by $\omega_v (A) := \scal{v}{Av} \ \forall A\in\elle{\kk_1\otimes\vv}$. Note that $\jj_{B\,\ast} : \ell^1 (X;\trcl{\vv}) \frecc \C$. Thus, if for all $i\in X$ we define the following normalized $\lv$-valued POVM on $\Omega$:
$$
Q_i :  \sigma (\Omega) \frecc  \lv \, , \qquad Q_{i, B} := (\jj_B)_i \, ,
$$
then we have
$$
\jj_{B\,\ast} (\delta_i \, \sigma) = \trt{\sigma Q_{i,B}} \quad \forall \sigma\in \trcl{\vv}
$$
and Eq.~\eqref{eq:non so2} becomes
$$
\SS_B (\ee) = \sum_{i\in X} \trq{Q_{i,B} (\ee\otimes \ii_\vv)_\ast (\omega_v)_i} \, ,
$$
which shows that, conditionally on the outcome $i\in X$, we just perform a measurement with POVM $Q_i$ on the states in $\trcl{\vv}$. 
%Now, in our case we have the identification  
%\begin{align} 
%\mm_1 \votimes \lv  \simeq   \bigoplus_{i\in  \mathsf X }     \elle {\vv_i} \, ,  \qquad \vv_i  \simeq  \vv \quad \forall  i \in \mathsf X.
%\end{align}  
%Hence,  the instrument   $\jj:  \sigma (\Omega) \frecc  \cp{\mathbb C,\mm_1 \votimes \lv}$  can be written in the form 
%\begin{align}
%\jj_B   (A)   =      \bigoplus_{i\in \mathsf X} \jj_{B,i}  (A)  \quad \forall A  \in  \mathbb C , 
%\end{align}
%where  each $\jj_i:  \sigma (\Omega) \frecc  \cp{\mathbb C, \lv}$   is a quantum instrument, i.e. $  \jj_{i, \Omega}  (  1)  =  I_{\vv}$ for every $i \in \mathsf X$.  In the Schr\"odinger picture, the instrument $\jj_*$  can be realized by first reading the classical information  carried by the system with algebra  $\mm_1$ and then performing the quantum instrument $\jj_{i*}$ which yields outcomes in $\Omega$ and transforms states in $\elle{ \vv}_*$  into probabilities (states on $\mathbb C$).
In other words, Theorem \ref{Osawa} claims that the most general way to extract information about a  measuring device on system $\kk_1$ consists in
\begin{enumerate} 
\item preparing a pure bipartite state $\omega_v$ in $\kk_1 \otimes \vv$;
\item performing the given measurement $\ee$ on $\kk_1$, thus obtaining the outcome $i\in X$;
\item conditionally on the outcome $i \in X$, performing a measurement (the POVM $Q_i$) on  the ancillary system $\vv$, thus obtaining an outcome in $\Omega$.
\end{enumerate}
Note that the choice of the POVM $Q_i$ depends in general on the outcome of the first measurement $\ee$. 

\subsubsection{Tranformations of measuring devices induced by a higher-order measurement}
In a quantum measurement it is often interesting to consider not only the statistics of the outcomes, but also how the measured object changes due to the measurement process.   
For example, in the case of ordinary quantum measurements, one is interested in studying the state reduction due to the occurrence of particular measurement outcomes.  
We can ask the same question in the case of higher-order measurements on quantum devices: for example, we can imagine a measurement process where a measuring device is tested, and, due to the test, is transformed into a new measuring device.   
This situation is described mathematically by a quantum superinstrument with outcomes in an outcome set $\Omega$, sending measurements in $\cp{ \mm_1, \elle{\kk_1} }$ to measurements in $\cp{  \mm_2, \elle{\kk_2}}$, where $\mm_1 \equiv \ell^\infty (X)$ and $\mm_2 \equiv \ell^\infty (Y)$ for some countable sets $X$ and $Y$.  

In this case, it follows from Theorem \ref{Osawa} that every superinstrument $\SS: \sigma (\Omega) \frecc \cpq{\ell^\infty (X) ,\elle{\kk_1}; \ell^\infty (Y) , \elle{\kk_2}}$ is of the form 
$$
[\SS_B (\ee)]  (f)  =  V^\ast [(\ee\otimes \ii_\vv)   \jj_B (f)] V \quad  \forall \ee\in\cb{\ell^\infty (X),\elle{\kk_1}} \, , \, \forall f\in\ell^\infty (Y)
$$
for all $B\in\sigma(\Omega)$, where $\vv$ is an ancillary Hilbert space, $V\in \elle{\kk_2,  \kk_1\otimes  \vv} $ is an isometry,   and $\jj:  \sigma (\Omega) \frecc  \cp{\ell^\infty (Y) , \ell^\infty (X;\lv)}$ is an instrument. Note that, by commutativity of $\ell^\infty (Y)$, the set $\cp{\ell^\infty (Y) , \ell^\infty (X;\lv)}$ coincides with the set of weak*-continuous {\em positive} maps from $\ell^\infty (Y)$ into $\ell^\infty (X;\lv)$. If for all $i\in X$ we define the positive map
$$
\jj_{i,B} : \ell^\infty (Y) \frecc \lv \, , \qquad \jj_{i,B} (f) := \jj_B (f)_i \, ,
$$
then each mapping $\jj_i :\sigma(\Omega)\frecc \cp{\ell^\infty (Y) , \lv}$ is an instrument, with predual
$$
\jj_{i,B \,\ast} : \trcl{\vv} \frecc \ell^1 (Y) \, , \qquad \jj_{i,B \,\ast} (\sigma) = \jj_{B\,\ast} (\delta_i\,\sigma)
$$
for all $B\in\sigma(\Omega)$. From the relation
$$
[\SS_B (\ee)]_\ast  (\rho)  =  [\jj_{B\,\ast} (\ee\otimes \ii_\vv)_\ast](V\rho V^\ast) = \sum_{i\in X} \jj_{i,B \,\ast} [(\ee\otimes \ii_\vv)_\ast (V\rho V^\ast)_i] \, ,
$$
holding for all states $\rho\in\trcl{\kk_2}$, we then see that the most general measurement on a quantum measuring device can be implemented by   
\begin{enumerate} 
\item applying an invertible dynamics (the isometry  $V)$ that transforms the input system $\kk_2$  into the composite system  $\kk_1  \otimes \vv$, where $\vv$ is an ancillary system;
\item performing the given measurement $\ee$ on $\kk_1$, thus obtaining the outcome $i\in X$;
\item conditionally to the outcome $i \in X$, performing a quantum measurement  (the predual instrument  $\jj_{i\,\ast}$), thus obtaining an outcome in $\Omega$ and transforming the ancillary system $\vv$ into the classical system described by the commutative algebra $\ell^\infty(Y)$. 
\end{enumerate}

If we assume that the set $\Omega$ also is countable, then the instrument $\jj:  \sigma (\Omega) \frecc  \cp{\ell^\infty(Y) , \ell^\infty(X,\lv)}$ is completely specified by its action on singleton sets, that is, by the countable set of quantum operations $\{\jj_{\omega}  \in \cp{\ell^\infty(Y) , \ell^\infty(X,\lv)} \mid \omega \in  \Omega\} $.  
In this case, if for all $i\in X$ we define
$$
Q^{(i)}_{\omega , j} := \jj_\omega (\delta_j)_i = \jj_{i,\omega} (\delta_j) \quad \forall (\omega , j)\in \Omega\times Y \, ,
$$
then the map $(\omega , j) \mapsto Q^{(i)}_{\omega , j}$ is a normalized POVM on the product set $\Omega\times Y$ and with values in $\lv$. Note that, in terms of the POVM $Q^{(i)}$, we can express each $\jj_{i,\omega}$ as
$$
\jj_{i,\omega}  (f)  =    \sum_{j \in Y}    f_j \, Q^{(i)}_{\omega,j} \quad \forall f\in\ell^\infty(Y)
$$
or, equivalently,
$$
(\jj_{i,\omega\,\ast}  (\sigma))_j  =  \trt{\sigma Q^{(i)}_{\omega,j}} \quad \forall \sigma\in\trcl{\vv} \, .  
$$ 
In other words, the step (3)  in the measurement process can be interpreted as a quantum measurement with outcome  $(\omega, j) \in \Omega \times Y$, where only the classical information concerning  the index $j \in Y$ is encoded in a physical system available for future experiments, whereas the information concerning index $\omega \in \Omega$  becomes unavailable after being read out by the experimenter.

\section*{Acknowledgements}
G.~C.~acknowledges support by the National Basic Research Program of China (973) 2011CBA00300 (2011CBA00301). A.~T.~and V.~U.~gratefully acknowledge the financial support of the Italian Ministry of Education, University and Research (FIRB project RBFR10COAQ).


\begin{thebibliography}{99}

\bibitem{Arveson} Arveson, W.~B., ``Subalgebras of C*-algebras'', Acta Math.~{\bf 123}, 141-224 (1969).
\bibitem{belavkin} Belavkin, V.~P., ``Reconstruction theorem for a quantum stochastic process'', Theor.~Math. Phys.~{\bf 62}, 275-289 (1985).
\bibitem{BelStasz} Belavkin, V.~P., and Staszewski, P., ``A Radon-Nikodym theorem for completely positive maps'', Rep.~Math.~Phys.~{\bf 24}, 49-55 (1986).
\bibitem{opttomo} Bisio, A.,  Chiribella, G., D'Ariano, G.~M., Facchini, S., and Perinotti, P.: ``Optimal quantum tomography for states, measurements, and transformations'', Phys.~Rev.~Lett.~{\bf 102}, 010404 (2009).
\bibitem{unitlearn}  Bisio, A.,  Chiribella, G., D'Ariano, G.~M., Facchini, S., and Perinotti, P.: ``Optimal quantum learning of a unitary transformation'', Phys.~Rev.~A {\bf 81}, 032324 (2010).
\bibitem{tredeoff}  Bisio, A.,  Chiribella, G., D'Ariano, G.~M., and Perinotti, P.: ``Information-disturbance tradeoff in estimating a unitary transformation'', Phys.~Rev.~A {\bf 82},   062305 (2010).
\bibitem{measclon} Bisio, A.,   D'Ariano, G.~.M., Perinotti, P., and Sedlak, M.: ``Cloning of a quantum measurement'', Phys.~Rev.~A {\bf 84}, 042330 (2011).
\bibitem{measlearn}   Bisio, A.,   D'Ariano, G.~.M., Perinotti, P., and Sedlak, M.: ``Quantum learning algorithms for quantum measurements'', Phys.~Lett.~A~{\bf 375}, 3425-3434 (2011).
\bibitem{BlM} Blecher, D.~P., and Le Merdy, C., {\em Operator Algebras and Their Modules} (Oxford University Press, Oxford, 2004).
\bibitem{BS} Blecher, D.~P., and Smith, R.~R.: ``The dual of the Haagerup tensor product'', J.~London Math.~Soc.~{\bf 45}, 126-144 (1992).
\bibitem{combs} Chiribella, G., D'Ariano, G.~M., and Perinotti, P., ``Quantum circuits architecture'',  Phys.~Rev.~Lett.~{\bf 101}, 060401 (2008).
\bibitem{memorydisc} Chiribella, G., D'Ariano, G.~M., and Perinotti, P.: ``Memory effects in quantum channel discrimination'', Phys.~Rev.~Lett.~{\bf 101}, 180501 (2008).
\bibitem{unitclon} Chiribella, G., D'Ariano, G.~M., and Perinotti, P.: ``Optimal cloning of a unitary transformation'', Phys.~Rev.~Lett.~{\bf 101}, 180504 (2008).
\bibitem{CDaP1} Chiribella, G., D'Ariano, G.~M., and Perinotti, P., ``Transforming quantum operations: quantum supermaps'', Europhys.~Lett.~{\bf 83}, 30004 (2008).
\bibitem{CDaP2} Chiribella, G., D'Ariano, G.~M., and Perinotti, P., ``Theoretical framework for quantum networks'', Phys.~Rev.~A {\bf 80}, 022339 (2009).
\bibitem{procqcmc} Chiribella, G., D'Ariano, G.~M., and Perinotti, P., ``Optimal covariant quantum networks'', AIP Conf.~Proc.~{\bf 1110}, 47-56 (2009).
\bibitem{bitcommitment}  Chiribella, G., D'Ariano, G.~M., and Perinotti, P.:  ``A short impossibility proof of quantum bit commitment'', arXiv:0905.3801 (2009).
\bibitem{puri} Chiribella, G., D'Ariano, G.~M., and Perinotti, P.: ``Probabilistic theories with purification'', Phys.~Rev.~A {\bf 81}, 062348 (2010).
\bibitem{dirk}  Chiribella, G., D'Ariano, G.~M., and Schlingemann, D.~M.: ``Barycentric decomposition of quantum measurements in finite dimensions'', J.~Math.~Phys.~{\bf 51}, 022111 (2010).
\bibitem{lor} D'Ariano, G.~M., Maccone, L., and Presti, P.~L.:  ``Quantum calibration of measurement instrumentation'', Phys.~Rev.~Lett.~{\bf 93}, 250407 (2004).
\bibitem{QTOS76}
Davies, E.~B., {\em Quantum Theory of Open Systems} (Academic Press, London, 1976).
\bibitem{DavLew}
Davies, E.~B., and Lewis, J.~T., ``An operational approach to quantum probability'',
Comm.~Math.~Phys.~{\bf 17}, 239-260 (1970).
\bibitem{EvLew}
Evans, D.~E., and Lewis, J.~T., {\em Dilations of Irreversible Evolutions in Algebraic Quantum Theory}, in Communications of the Dublin Institute for Advanced Studies, Series A (Theoretical Physics), No.~24 (Dublin, 1977).
\bibitem{fiur} Fiurasek, J.: ``Maximum-likelihood estimation of quantum measurement'', Phys.~Rev.~A {\bf 64}, 024102 (2001).
\bibitem{Haag} Haagerup, U., ``Decomposition of completely bounded maps on operator algebras'', unpublished preprint.
\bibitem{Holevo01} Holevo, A.~S., {\em Statistical Structure of Quantum Theory} (Springer, Berlin, 2001).
\bibitem{jencova} Jencova, A., ``Generalized channels: Channels for convex subsets of the state space'', J.~Math.~Phys.~{\bf 53}, 012201 (2012). 
\bibitem{KadRinI} Kadison, R.~V., and Ringrose, J.~R., {\em Fundamentals of the Theory of Operator Algebras, Volume I: Elementary Theory} (Academic Press, New York, 1983).
\bibitem{KadRin} Kadison, R.~V., and Ringrose, J.~R., {\em Fundamentals of the Theory of Operator Algebras, Volume II: Advanced Theory} (Academic Press, Orlando, 1986).
\bibitem{Kraus71} Kraus, K., ``General state changes in quantum theory'', Ann.~Phys.~{\bf 64}, 311-335 (1971).
\bibitem{lindblad} Lindblad, G., ``Non-Markovian quantum stochastic processes and their entropy'', Comm. Math.~Phys.~{\bf 65}, 281-294 (1979).
\bibitem{soto} Luis, A., and Sanchez-Soto, L.~L.: ``Complete characterization of arbitrary quantum measurement processes'', Phys.~Rev.~Lett.~{\bf 83}, 3573-3576 (1999).
\bibitem{eisert} Lundeen, J.~S.,  Feito, A.,  Coldenstrodt-Ronge, H., Pregnell,  K~.L.,  Silberhorn,  C., Ralph,    T.~C.,  Eisert,  J., Plenio,  M.~B.,  and Walmsley, I.~A.:  ``Tomography of quantum detectors'', Nat.~Phys.~{\bf 5}, 27-30 (2009).
\bibitem{Ozawa}
Ozawa, M., ``Quantum measuring processes of continuous observables'', J.~Math.~Phys.~{\bf 25}, 79-87 (1984).
\bibitem{parthasarathy}  Parthasarathy, K.~R., ``A continuous time version of Stinespring's theorem on completely positive maps'', in {\em Quantum probability and applications V}, Lecture Notes in Mathematics 1442 (Springer, Berlin / Heidelberg, 1990), pp.~296-300.   
\bibitem{Paul} Paulsen, V., {\em Completely bounded maps and operator algebras} (Cambridge University Press, Cambridge, 2002).
\bibitem{Raginsky} Raginsky, M., ``Radon-Nikodym derivatives of quantum operations'', J.~Math.~Phys.~{\bf 44}, 5003-5019 (2003).
\bibitem{Sakai} Sakai, S., {\em $C^\ast$-algebras and $W^\ast$-algebras} (Springer, Berlin, 1971).
\bibitem{AJP} Rebolledo, R., ``Complete positivity and the Markov structure of open quantum systems'', in {\em Open Quantum Systems II - The Markovian Approach}, Lecture Notes in Mathematics 1881 (Springer, Berlin, 2006), pp.~149-182.
\bibitem{SedlakZiman} Sedlak, M., and Ziman, M., ``Unambiguous comparison of unitary channels'', Phys.~Rev.~A {\bf 79}, 012303 (2009).
\bibitem{Stine} Stinespring, W.~F., ``Positive Functions on $C^\ast$-Algebras'', Proc.~Amer.~Math.~Soc.~{\bf 6}, 211-216 (1955).
\bibitem{Tak} Takesaki, M., {\em Theory of Operator Algebra, Volume I} (Springer, Berlin, 1979). 
\bibitem{Ziman} Ziman, M., ``Process POVM: A mathematical framework for the description of process tomography experiments'', Phys.~Rev.~A {\bf 77}, 062112 (2008).
\bibitem{zycko} Zyczkowski, K., ``Quartic quantum theory: an extension of the standard quantum mechanics'', J.~Phys.~A {\bf 41}, 355302 (2008). 

\end{thebibliography}
\end{document}